\documentclass[11pt,a4paper]{amsart}
\usepackage{amsmath,amsfonts,amsthm,amssymb}
\usepackage{epsfig}
\usepackage{color}
\usepackage{graphicx}
\usepackage{subcaption}
\usepackage{cite,url}
\usepackage{hyperref}
\usepackage{xspace}
\usepackage{enumerate}
\usepackage[algoruled,vlined,english]{algorithm2e} 
\usepackage[left=1in,right=1in,top=1in,bottom=1in,foot=0.5in]{geometry}

\newtheorem{theorem}{Theorem}[section]
\newtheorem{lemma}[theorem]{Lemma}
\newtheorem{proposition}[theorem]{Proposition}
\newtheorem{corollary}[theorem]{Corollary}
\newtheorem{claim}{Claim}
\newtheorem{case}{Case}
\theoremstyle{definition}

\newtheorem{remark}[theorem]{Remark}

\DeclareMathOperator{\lca}{lca}

\newcommand{\N}{\ensuremath{\mathbb{N}}\xspace}
\newcommand{\cT}{\ensuremath{\mathcal T}\xspace}
\newcommand{\wrho}{\ensuremath{\widehat{\rho}}\xspace}

\newcommand{\commentout}[1]{}

\title{Fast approximation and exact computation of negative curvature
  parameters of graphs$^*$}\thanks{$^*$An extended
  abstract~\cite{CCDDMV-court} of this paper has appeared in the
  proceedings of SoCG 2018.}

\author[J.\ Chalopin]{J\' er\'emie Chalopin}
\address{CNRS, Aix-Marseille Universit\'e, Universit\'e de Toulon, LIS, Marseille, France}
\email{jeremie.chalopin@lis-lab.fr}

\author[V.\ Chepoi]{Victor Chepoi}
\address{Aix-Marseille Universit\'e, CNRS, Universit\'e de Toulon, LIS, Marseille, France}
\email{victor.chepoi@lis-lab.fr}

\author{Feodor F. Dragan}
\address{Computer Science Department, Kent State University, Kent, USA}
\email{dragan@cs.kent.edu}

\author{Guillaume Ducoffe}
\address{National Institute for Research and
  Development in Informatics and Research Institute of the University
  of Bucharest, Bucure\c{s}ti, Rom\^{a}nia}
\email{guillaume.ducoffe@ici.ro}

\author{Abdulhakeem Mohammed}
\address{Computer Science Department, Kent State University, Kent, USA}
\email{amohamm4@kent.edu}

\author[Y.\ Vax\`es]{Yann Vax\`es}
\address{Aix-Marseille Universit\'e, CNRS, Universit\'e de Toulon, LIS, Marseille, France}
\email{yann.vaxes@lis-lab.fr}

\begin{document}

\begin{abstract}
  In this paper, we study Gromov hyperbolicity and related parameters,
  that represent how close (locally) a metric space is to a tree from
  a metric point of view.  The study of Gromov hyperbolicity for
  geodesic metric spaces can be reduced to the study of {\em graph
    hyperbolicity}.  The main contribution of this paper is a new
  characterization of the hyperbolicity of graphs, via a new parameter
  which we call \emph{rooted insize}.  This characterization has
  algorithmic implications in the field of large-scale network
  analysis.  A sharp estimate of graph hyperbolicity is useful,
  {e.g.}, in embedding an undirected graph into hyperbolic space with
  minimum distortion [Verbeek and Suri, SoCG'14].  The hyperbolicity
  of a graph can be computed in polynomial-time, however it is
  unlikely that it can be done in \emph{subcubic} time.  This makes
  this parameter difficult to compute or to approximate on large
  graphs. Using our new characterization of graph hyperbolicity, we
  provide a simple factor 8 approximation algorithm (with an additive
  constant 1) for computing the hyperbolicity of an $n$-vertex graph
  $G=(V,E)$ in optimal time $O(n^2)$ (assuming that the input is the
  distance matrix of the graph).  This algorithm leads to constant
  factor approximations of other graph-parameters related to
  hyperbolicity (thinness, slimness, and insize). We also present the
  first efficient algorithms for exact computation of these
  parameters.  All of our algorithms can be used to approximate the
  hyperbolicity of a geodesic metric space.

  We also show that a similar characterization of hyperbolicity holds
  for all geodesic metric spaces endowed with a geodesic spanning
  tree. 
  Along the way, we prove that any complete geodesic metric space
  $(X,d)$ has such a geodesic spanning tree.
\end{abstract}

\maketitle

\section{Introduction}\label{intro}

Understanding the geometric properties of complex networks is a key issue in network analysis and geometric graph theory.
One important such property is  {\em negative curvature}~\cite{NaSa}, causing the traffic between the vertices
to pass through a relatively small core of the network -- as if the shortest paths between them were curved inwards.
It has been empirically observed, then formally proved~\cite{ChDrVa17}, that such a phenomenon is related to the value
of the {\em Gromov hyperbolicity} of the graph. In this paper, we propose exact and approximation algorithms to compute
hyperbolicity of a graph and its relatives (the approximation algorithms can be applied to geodesic metric spaces as well).

A metric space $(X,d)$ is $\delta$-\emph{hyperbolic}
\cite{AlBrCoFeLuMiShSh,BrHa,Gr} if for any four points $w,v,x,y$ of
$X$, the two largest of the distance sums $d(w,v)+d(x,y)$,
$d(w,x)+d(v,y)$, $d(w,y)+d(v,x)$ differ by at most $2\delta \geq 0$. A
graph $G=(V,E)$ endowed with its standard graph-distance $d_G$
is $\delta$-\emph{hyperbolic} if the metric space $(V,d_G)$ is
$\delta$-{hyperbolic}.   In case of geodesic metric spaces and graphs,
$\delta$-hyperbolicity can be defined in other equivalent
ways, {e.g.}, via the thinness, slimness, or insize of geodesic triangles.
The hyperbolicity $\delta(X)$ of a metric space $X$ is the smallest
$\delta\ge 0$ such that $X$ is $\delta$-hyperbolic. It can be viewed as a local measure of how close $X$ is to a
tree: the smaller the hyperbolicity is, the closer the metrics of its
$4$-point subspaces are close to tree-metrics.

The study of hyperbolicity of graphs is motivated by the fact that many real-world graphs are tree-like from a
metric point of view \cite{AADr,AdcockSM13,Bo++} or have small hyperbolicity \cite{KeSN16,NaSa,ShavittT08}. This is due to the fact that
many of these graphs (including Internet application networks, web networks, collaboration networks, social networks, biological networks,
and others) possess certain geometric and topological characteristics. Hence, for many applications,  including the design of efficient
algorithms (cf., {e.g.}, \cite{Bo++,ChChPaPe,Chepoi08,Chepoi12,ChDrVa17,Chepoi07,DGKMY2015,EdKeSa18,VS2014}), it is useful to know an accurate
approximation of the hyperbolicity $\delta(G)$ of a graph $G$.

\subsection*{Related work.}
For an $n$-vertex graph $G$, the definition of hyperbolicity directly implies a simple brute-force $O(n^4)$ algorithm to compute $\delta(G)$.
This running time is too slow for computing the hyperbolicity of large graphs that occur in
applications \cite{AADr,Bo++,BoCrHa,FouIsVi}. On the theoretical side, it was shown that relying on 
matrix multiplication results, one can
improve the upper bound on time-complexity to $O(n^{3.69})$ \cite{FouIsVi}. Moreover, roughly quadratic lower
bounds are known \cite{BoCrHa,CoDu,FouIsVi}. In practice, however, the best known algorithm still has
an $O(n^4)$-time worst-case bound but uses several clever tricks when compared to
the brute-force algorithm \cite{Bo++}. Based on empirical studies,
an $O(mn)$ running time is claimed, where $m$ is the number of edges in the
graph. Furthermore, there are heuristics for computing the hyperbolicity of a
given graph \cite{CoCoLa15}, and there are investigations of whether one can compute hyperbolicity
in linear time when some graph parameters 
take small values \cite{Rolf+++,CoDuPo}.

Perhaps it is interesting to notice that the first algorithms for computing the Gromov hyperbolicity were designed for
Cayley graphs of finitely generated groups (these are infinite vertex-transitive graphs of uniformly bounded degrees).
Gromov gave an algorithm to recognize Cayley graphs of hyperbolic
groups and estimate the hyperbolicity constant $\delta $. His
algorithm is based on the theorem that in Cayley graphs, the hyperbolicity ``propagates'', {i.e.}, if balls of an appropriate fixed radius induce a $\delta$-hyperbolic space, then the whole space is $\delta'$-hyperbolic for some
$\delta'>\delta $ (see \cite{Gr}, 6.6.F and \cite {DelGro}). Therefore,
in order to compute the hyperbolicity of a Cayley graph, it is
enough to verify the hyperbolicity of a sufficiently big ball (all balls of a given radius in a Cayley graph are isomorphic to
each other). For other algorithms deciding if the Cayley graph of a finitely
generated group is  hyperbolic, see \cite{Bowd,Pa2}. However, similar methods
do not help when dealing with arbitrary graphs.

By a result of Gromov \cite{Gr}, if the four-point condition in the definition of
hyperbolicity holds for a fixed basepoint $w$ and any triplet $x,y,v$
of $X$, then the metric space $(X,d)$ is $2\delta$-hyperbolic. This
provides a factor 2 approximation of hyperbolicity of a metric space
on $n$ points running in cubic $O(n^3)$ time. Using fast algorithms
for computing (max,min)-matrix products, it was noticed in
\cite{FouIsVi} that this 2-approximation of hyperbolicity can be
implemented in $O(n^{2.69})$ time. In the same paper, it was shown that
any algorithm computing the hyperbolicity for a fixed basepoint in
time $O(n^{2.05})$ would provide an algorithm for $(\max,\min)$-matrix
multiplication faster than the existing ones.  In~\cite{Du},
approximation algorithms are given to compute a
$(1+\epsilon)$-approximation in $O(\epsilon^{-1}n^{3.38})$ time and a
$(2+\epsilon)$-approximation in $O(\epsilon^{-1}n^{2.38})$ time.
As a direct application of the characterization of hyperbolicity of
graphs via a cop and robber game and dismantlability,
\cite{ChChPaPe} presents a simple constant factor approximation
algorithm for hyperbolicity of $G$ running
in optimal $O(n^2)$ time. Its approximation
ratio is huge (1569), however it is believed that its theoretical performance
is much better and the factor of 1569 is mainly due to the use in the proof
of the definition of hyperbolicity via linear isoperimetric inequality.
This shows that the question
of designing fast and (theoretically certified) accurate algorithms for approximating
graph hyperbolicity is still an important and open question.

\subsection*{Our contribution.} In this paper, we tackle this open question and propose a very simple (and thus practical)  factor 8
algorithm for approximating the hyperbolicity $\delta(G)$ of an
$n$-vertex graph $G$ running in optimal $O(n^2)$ time. As in several
previous algorithms, we assume that the input is the distance matrix
$D$ of the graph $G$.
Our algorithm picks a basepoint $w$, a Breadth-First-Search tree $T$
rooted at $w$, and considers only geodesic triangles of $G$ with one
vertex at $w$ and two sides on $T$. For all such sides in $T$, it
computes the maximum over all distances between the two preimages of
the centers of the respective tripods (see Section \ref{hyperb} for
definitions). This maximum $\rho_{w,T}(G)$ (called \emph{rooted
  insize}) can be easily computed in $O(n^2)$ time and, as we
demonstrate, provides an 8-approximation (with an additive constant 1)
for $\delta(G)$. If the graph $G$ is given by its adjacency list, then
we show that $\rho_{w,T}(G)$ can be computed in $O(nm)$ time and
linear $O(n+m)$ space.  For geodesic spaces $(X,d)$ endowed with a
geodesic spanning tree we show that we can also define the rooted
insize $\rho_{w,T}(X)$ and that the same relationships between
$\rho_{w,T}(X)$ and the hyperbolicity $\delta(X)$ hold, thus providing
a new characterization of hyperbolicity. \emph{En passant}, we show
that any complete geodesic space $(X,d)$ always has such a geodesic
spanning tree (this result is not trivial, see the proof of
Theorem~\ref{BFS_geodesic} and Remark~\ref{rem:GS-not-BFS}). We hope
that this fundamental result can be useful in other contexts.

Perhaps it is surprising that hyperbolicity that is originally defined
via quadruplets and can be 2-approximated via triplets ({i.e.}, via
pointed hyperbolicity), can be finally defined and approximated only
via pairs (and an arbitrary fixed BFS-tree). Indeed, summarizing our
contributions, we proved the existence of some property
$P_{w,T}(x,y : \delta)$, defined w.r.t. a fixed basepoint $w$ and a
fixed BFS tree $T$, such that: \texttt{(i)} for any
$\delta$-hyperbolic graph the property holds for any pair $x,y$ of
vertices; and conversely \texttt{(ii)} if the property holds for every
pair $x,y$ then the graph is $8\delta$-hyperbolic.  See Theorem
\ref{main} for more details. We hope that this new characterization
can be useful in establishing that graphs and simplicial complexes
occurring in geometry and in network analysis are hyperbolic.

The way the rooted insize $\rho_{w,T}(G)$ is computed is closely
related to how hyperbolicity is defined via slimness, thinness, and
insize of its geodesic triangles. Similarly to the hyperbolicity
$\delta(G)$, one can define slimness $\varsigma(G)$, thinness
$\tau(G)$, and insize $\iota(G)$ of a graph $G$. As a direct
consequence of our algorithm for approximating $\delta(G)$ and the
relationships between $\delta(G)$ and $\varsigma(G),\tau(G),\iota(G)$,
we obtain constant factor $O(n^2)$ time algorithms for approximating
these parameters.  On the other hand, an {\em exact} computation, in
polynomial time, of these geometric parameters has never been
provided.  In Theorem~\ref{thm:slim-thin-exact}, we show that the
thinness $\tau(G)$ and the insize $\iota(G)$ of a graph $G$ can be
computed in $O(n^2m)$ time and the slimness $\varsigma(G)$ of $G$ can
be computed in $\widehat{O}(n^2m + n^4/\log^3n)$ time\footnote{The
  $\widehat{O}(\cdot)$ notation hides polyloglog factors.}
combinatorially and in $O(n^{3.273})$ time using matrix
multiplication.  However, we show that the minimum value of
$\rho_{w,T}(G)$ over all basepoints $w$ and all BFS-trees $T$ cannot
be approximated in polynomial time with a factor strictly better than
2 unless P = NP.

The new notion of rooted insize, as well as the classical notions of  thinness,
slimness, and insize can be defined only for unweighted graphs and
geodesic metric spaces. Therefore, the approximation of hyperbolicity
via the rooted insize (and the corresponding algorithms) do not hold
for arbitrary metric spaces (such as weighted graphs for example).

\section{Gromov hyperbolicity and its relatives}\label{hyperb}

\subsection{Gromov hyperbolicity}
Let $(X,d)$ be a metric space and $w\in X$. The \emph{Gromov product}\footnote{Informally, $(y|z)_w$ can be viewed as half the detour you make, when going over $w$ to get from $y$ to $z.$} of $y,z\in X$ with respect to $w$ is
$(y|z)_w=\frac{1}{2}(d(y,w)+d(z,w)-d(y,z)).$
A metric space $(X,d)$ is  $\delta$-\emph{hyperbolic} \cite{Gr} for $\delta\ge 0$ if
$(x|y)_w\ge \min \{ (x|z)_w, (y|z)_w\}-\delta$
for all $w,x,y,z\in X$. Equivalently, $(X,d)$ is $\delta$-hyperbolic
if  for any $u,v,x,y\in X$, the two largest of the sums
$d(u,v)+d(x,y)$, $d(u,x)+d(v,y)$, $d(u,y)+d(v,x)$ differ by at most
$2\delta \geq 0$. A metric space $(X,d)$ is said to be $\delta$-\emph{hyperbolic with respect to a basepoint $w$}
if $(x|y)_w\ge \min \{ (x|z)_w, (y|z)_w\}-\delta$ for all $x,y,z\in X$. 

\begin{proposition} \cite{AlBrCoFeLuMiShSh,BrHa,Gr,GhHa}\label{hyp_basepoint}
If $(X,d)$ is $\delta$-hyperbolic with respect to some basepoint, then $(X,d)$ is $2\delta$-hyperbolic.
\end{proposition}

Let $(X,d)$ be a metric space.  An $(x,y)$-\emph{geodesic}
is a (continuous) map $\gamma: [0,d(x,y)] \to X$ from the segment $[0,d(x,y)]$
of ${\mathbb R}^1$ to $X$ such that
$\gamma(0)=x, \gamma(d(x,y))=y,$ and $d(\gamma(s),\gamma(t))=|s-t|$ for all $s,t\in
[0,d(x,y)].$ A \emph{geodesic segment} with endpoints $x$ and $y$ is the image
of the map $\gamma$  (when it is clear from the context, by a geodesic we
mean a geodesic segment and we denote it by $[x,y]$). A metric space $(X,d)$ is
\emph{geodesic} if every pair of
points in $X$ can be joined by a geodesic.
A  \emph{real tree} (or an ${\mathbb R}$-\emph{tree}) \cite[p.186]{BrHa} is a geodesic metric space $(T,d)$ such that
\begin{enumerate}[(1)]
\item there is a unique geodesic $[x,y]$ joining each pair of points $x,y\in T$;
\item if $[y,x]\cap [x,z]=\{ x\}$, then $[y,x]\cup [x,z]=[y,z].$
\end{enumerate}

Let $(X,d)$ be a geodesic metric space.  A \textit{geodesic triangle}
$\Delta(x,y,z)$ with $x, y, z \in X$ is the union $[x,y] \cup [x,z]
\cup [y,z]$ of three geodesics connecting these points.  A
geodesic triangle $\Delta(x,y,z)$ is called $\delta$-\emph{slim} if for
any point $u$ on the side $[x,y]$ the distance from $u$ to $[x,z]\cup
[z,y]$ is at most $\delta$. Let $m_x$ be the point of
$[y,z]$ located at distance $\alpha_y :=(x|z)_y$ from $y.$ Then, $m_x$ is located at
distance $\alpha_z :=(y|x)_z$ from $z$ because
$\alpha_y + \alpha_z = d(y,z)$. Analogously, define the points
$m_y\in [x,z]$ and $m_z\in [x,y]$ both located at distance $\alpha_x
:=(y|z)_x$ from $x;$ see Fig.~\ref{fig1} for an
illustration.
We define a tripod $T(x,y,z)$ consisting of three solid segments
$[x,m],[y,m],$ and $[z,m]$ of lengths $\alpha_x,\alpha_y,$ and
$\alpha_z,$ respectively. The function mapping the vertices $x,y,z$ of
$\Delta(x,y,z)$ to the respective leaves of $T(x,y,z)$ extends
uniquely to a function $\varphi: \Delta(x,y,z) \to T(x,y,z)$ such that the
restriction of $\varphi$ on each side of $\Delta(x,y,z)$ is an
isometry. This function maps the points $m_x,m_y,$
and $m_z$ to the center $m$ of $T(x,y,z)$. Any other point  of
$T(x,y,z)$ is the image of at most two points of $\Delta
(x,y,z)$.
A geodesic triangle $\Delta(x,y,z)$ is called
$\delta$-\emph{thin} if for all points $u,v\in \Delta(x,y,z),$
$\varphi(u)=\varphi(v)$ implies $d(u,v)\le \delta.$ The \emph{insize} of
$\Delta(x,y,z)$ is the diameter of the preimage $\{ m_x,m_y,m_z\}$ of the center
$m$ of the tripod $T(x,y,z)$. Below, we remind that the hyperbolicity of a
geodesic space can be approximated by the maximum thinness and slimness of its geodesic triangles.

\begin{figure}
\begin{center}
\scalebox{1}{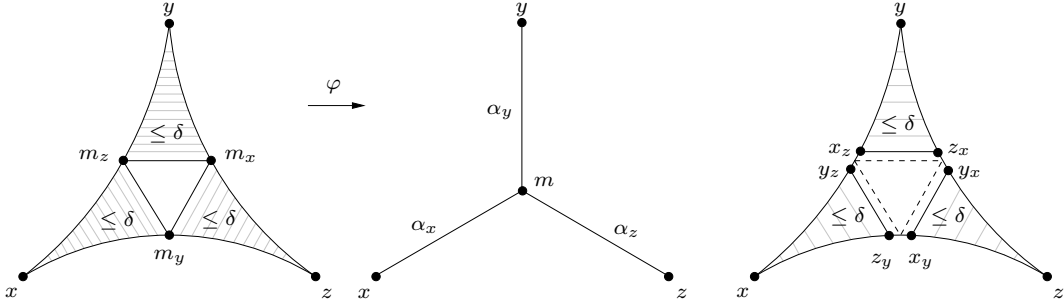}
\end{center}
\caption{Insize and thinness in geodesic spaces and graphs.} \label{fig1}
\end{figure}

For a geodesic  metric space $(X,d)$, one can define the following parameters: 
\begin{itemize}
\item \emph{hyperbolicity} $\delta(X)=\min\{ \delta: X \text{ is } \delta\text{-hyperbolic}\},$
\item \emph{pointed hyperbolicity} $\delta_w(X)=\min \{ \delta:  X \text{ is } \delta\text{-hyperbolic  with respect to a basepoint } w\},$
\item \emph{slimness} $\varsigma(X)=\min\{ \delta: \text{any geodesic triangle of } X \text{ is } \delta\text{-slim}\},$
\item \emph{thinness} $\tau(X)=\min\{ \delta: \text{any geodesic triangle of } X \text{ is } \delta\text{-thin}\},$
\item \emph{insize} $\iota(X)=\min\{ \delta: \text{the insize of any geodesic triangle of } X \text{ is at most } \delta\}.$
\end{itemize}

\begin{proposition} \cite{AlBrCoFeLuMiShSh,BrHa,Gr,GhHa,Soto}\label{hyp_charact}
For a geodesic metric space $(X,d)$, $\delta(X) \leq
\iota(X)=\tau(X)\le 4\delta(X)$, $\varsigma(X) \leq \tau(X) \leq
4\varsigma(X)$, and $\delta(X) \le 2\varsigma(X) \le 6\delta(X)$.
\end{proposition}

Due to Propositions \ref{hyp_basepoint} and \ref{hyp_charact}, a
geodesic metric space $(X,d)$ is called \emph{hyperbolic} if one of the
numbers $\delta(X),\delta_w(X),\varsigma(X),\tau(X),\iota(X)$ (and
thus all) is finite. Notice also that a geodesic metric space $(X,d)$
is 0-hyperbolic if and only if $(X,d)$ is a real tree
\cite[p.399]{BrHa} (and in this case,
$\varsigma(X)=\tau(X)=\iota(X)=\delta(X)=0$).

\subsection{Hyperbolicity of graphs} \label{ss:HypOfGraphs}

All graphs $G=(V,E)$ occurring in this paper are undirected and
connected, but not necessarily finite (in algorithmic results they
will be supposed to be finite). For a vertex $v \in V$, we denote by
$N_G(v)$ the open neighborhood of $v$, by $N_G[v]$ the closed
neighborhood of $v$, and by $\deg_G(v)$ the degree of $v$ (when $G$ is
clear from the context, the subscripts will be omitted).
For any two vertices $x,y\in V,$ the \emph{distance}
$d(x,y)$ is the minimum number of edges in a path between $x$ and $y.$ Let $[x,y]$
denote a shortest path connecting vertices $x$ and $y$ in $G$; we  call $[x,y]$
a \emph{geodesic} between $x$ and $y$. The \emph{interval}
$I(u,v)=\{ x\in V: d(u,x)+d(x,v)=d(u,v)\}$  consists of all vertices on
$(u,v)$-geodesics. 
There is a strong analogy between the metric properties of graphs and
geodesic metric spaces, due to their uniform local structure. Any
graph $G=(V,E)$ gives rise to a geodesic space $(X_G,d)$
(into which $G$ isometrically embeds) obtained by replacing each edge $xy$
of $G$ by a segment isometric to $[0,1]$ with ends at $x$ and
$y$. $X_G$ is called a \emph{metric graph}.  
Conversely, by \cite[Proposition 8.45]{BrHa}, any geodesic metric
space $(X,d)$ is (3,1)-quasi-isometric to a graph $G=(V,E)$. 
This graph $G$ is constructed in the following way: let $V$ be an open
maximal $\frac{1}{3}$-packing of $X$, {i.e.}, $d(x,y)>\frac{1}{3}$ for
any $x,y\in V$ (that exists by Zorn's lemma). Then two points $x,y\in
V$ are adjacent in $G$ if and only if $d(x,y)\le 1$. Since
hyperbolicity is preserved (up to a constant factor) by
quasi-isometries, this reduces the computation of hyperbolicity for
geodesic spaces to the case of graphs.

The notions of geodesic triangles, insize, $\delta$-slim and $\delta$-thin
triangles can also be defined in case of graphs with the single difference that for graphs, the center
of the tripod is not necessarily the image of any vertex on the
sides of $\Delta(x,y,z).$ 
For graphs, we ``discretize'' the notion of $\delta$-thin triangles in the following way.
We say that a geodesic triangle  $\Delta(x,y,z)$ of a graph $G$ is {\em $\delta$-thin}
if for any $v\in \{x,y,z\}$ and vertices $a\in [v,u]$ and
$b\in [v,w]$ ($u,w\in \{x,y,z\}$, and $u,v,w$ are distinct), $d(v,a)=d(v,b)\leq (u|w)_v$ implies $d(a,b)\leq \delta$.
A graph $G$ is {\em $\delta$-thin}, if  all geodesic triangles in $G$ are $\delta$-thin.
Given a  geodesic triangle $\Delta(x,y,z):=[x,y]\cup [x,z]\cup [y,z]$ in $G$, let $x_y$ and $y_x$
be the vertices of $[z,x]$ and $[z,y]$, respectively, both at distance $\lfloor (x|y)_z \rfloor$ from $z$.
Similarly, one can define vertices $x_z,z_x$ and vertices $y_z,z_y;$ see Fig.~\ref{fig1}.
The {\em insize} of $\Delta(x,y,z)$ is defined as $\max\{d(y_z,z_y), d(x_y,y_x), d(x_z,z_x)\}$.
An interval $I(x,y)$ is said to be $\kappa$-thin if $d(a,b) \le  \kappa$ for all $a,b \in I(x,y)$ with $d(x,a)=d(x,b).$
The smallest $\kappa$ for which all intervals of $G$ are $\kappa$-thin is called the {\em interval thinness} of $G$ and denoted by  $\kappa(G)$.
Denote also by  $\delta(G)$,  $\delta_w(G)$,
$\varsigma(G)$,  $\tau(G)$, and $\iota(G)$ respectively the hyperbolicity, the pointed hyperbolicity with respect to a basepoint $w$,
the slimness, the thinness, and the insize of a graph $G$.

\section{Auxiliary results}

We will need the following inequalities between $\varsigma(G)$,
$\tau(G)$, $\iota(G)$, and $\delta(G)$. They are known to be true for
all geodesic spaces (see \cite{AlBrCoFeLuMiShSh,BrHa,Gr,GhHa,Soto}). We
 present graph-theoretic proofs in case of graphs for
completeness (and due to slight modifications in their definitions for
graphs).

\begin{proposition} \label{prop:sl-vs-th}
  $\delta(G)-\frac{1}{2} \leq \iota(G)=\tau(G)\le 4\delta(G)$,
  $\varsigma(G) \leq \tau(G) \leq 4\varsigma(G)$, $\delta(G)
  -\frac{1}{2} \le 2\varsigma(G) \leq 6\delta(G)+1$, and
  $\kappa(G)\leq \min\{\tau(G),2 \delta(G), 2 \varsigma(G)\}$.
\end{proposition}

The fact that $\delta(G) \le 2\varsigma(G) + \frac{1}{2}$ is a result
of Soto~\cite[Proposition II.20]{Soto}.  For our convenience, we
reformulate and prove the other results in four lemmas, plus one
auxiliary lemma.

\begin{lemma} \label{lm:sl-vs-th1} $\varsigma(G)\leq \iota(G)=\tau(G)\leq 4 \varsigma(G)$.
\end{lemma}

\begin{proof} By the definitions of $\varsigma(G)$, $\tau(G)$, and $\iota(G)$, we only need to show that $\tau(G)\leq \iota(G)\leq 4 \varsigma(G)$.

Let $\iota:=\iota(G)$. Pick an arbitrary geodesic triangle $\Delta(x,y,z)$ of $G$ formed by shortest paths $[x,y]$, $[x,z]$, and $[y,z]$.
By induction on $k:=d(x,y)+d(x,z)$, we  show that $d(a,b)\leq \iota$ holds for every pair of vertices $a\in [x,y], b\in [x,z]$ with
$d(x,a)=d(x,b)\leq (y|z)_x$. Let $y'$ be the neighbor of $y$ on $[x,y]$. Consider a geodesic triangle $\Delta(x,y',z)$ formed by shortest
paths $[x,y']:=[x,y]\setminus \{y\}$, $[x,z]$ and $[y',z]$, where $[y',z]$ is an arbitrary shortest path connecting $y'$ with $z$.
Since $d(y,z)-1\leq d(y',z)\leq d(y,z)+1$, we have  $(y'|z)_x= (y|z)_x - \alpha$, where $\alpha\in \{0, \frac{1}{2}, 1\}$. Now,
for every pair of vertices $a\in [x,y'], b\in [x,z]$ with
$d(x,a)=d(x,b)\leq (y'|z)_x$, $d(a,b)\leq \iota$ holds by induction. If a pair $a\in [x,y], b\in [x,z]$ exists such that
$(y'|z)_x < d(x,a)=d(x,b)\leq (y|z)_x$, then $d(x,a)=d(x,b)= \lfloor(y|z)_x\rfloor$ and, therefore, $d(a,b)\leq \iota$ holds since the insize of $\Delta(x,y,z)$ is at most $\iota$.
Thus, we conclude that $\tau(G)\leq \iota(G)$.

Let $\varsigma:=\varsigma(G)$. Pick any geodesic triangle $\Delta(x,y,z)$ of $G$ formed by shortest paths
$[x,y]$, $[x,z]$, and $[y,z]$. Consider the vertices $x_y,y_x,y_z,z_y,x_z,z_x$ as defined in Subsection~\ref{ss:HypOfGraphs}.
It suffices to show that $d(y_z,z_y)\leq 4\varsigma$. Since $\varsigma(G)=\varsigma$, there is a vertex $a\in [x,z]\cup [y,z]$
such that $d(a,y_z)\leq \varsigma$. Assume $a\in [x,z]$. We claim that $d(y_z,z_y)\leq 2\varsigma$. Indeed, if $d(x,a)\leq d(x,z_y)$, then $d(x,y_z)=d(x,z_y)=d(x,a)+d(a,z_y)$ and $d(x,y_z)\le d(x,a)+d(a,y_z)\le d(x,a)+\varsigma$ imply $d(a,z_y)\leq \varsigma$ and hence $d(y_z,z_y)\leq 2\varsigma$. If $d(x,a)\geq d(x,z_y)$, then $d(x,z_y)+d(z_y,a)=d(x,a)\le d(x,y_z)+d(z_y,a)\le d(x,y_z)+ \varsigma$  implies $d(z_y,a)\le \varsigma$ and hence $d(y_z,z_y)\leq 2\varsigma$.

So, we may assume that $a$ belongs to $[y,z]$. If $a\in [y_x,z]\subseteq [y,z]$, then
$d(x,y_z)+d(z,z_x)=d(x,z_y)+d(z,z_y)=d(x,z)\leq d(x,y_z)+d(y_z,a)+ d(a,z)= d(x,y_z)+d(y_z,a)+d(z,z_x)-d(a,z_x)$. It implies that
$d(a,z_x)\leq d(y_z,a)\leq \varsigma$, yielding $d(y_z,y_x)\leq 2\varsigma$ and $d(y_z,z_x)\leq 2\varsigma$.
If $a\in [y,z_x]\subseteq [z,y]$, then  $d(y,a)+d(a,y_x)=d(y,y_x)=d(y,y_z)\leq d(y_z,a)+d(y,a)$, implying $d(a,y_x)\leq d(y_z,a)\leq \varsigma$.
Hence, $d(y_z,y_x)\leq 2\varsigma$ and $d(y_z,z_x)\leq 2\varsigma$.

By symmetry, also for vertex $z_y$, we can get $d(z_y,y_z)\leq 2\varsigma$ or $d(z_y,y_x)\leq 2\varsigma$. Therefore,
if $d(z_y,y_z)> 2\varsigma$, then $d(z_y,y_z)\leq d(z_y,y_x)+d(y_z,y_x)\leq 4\varsigma$ must hold. Thus, $\iota(G)\leq 4 \varsigma(G)$.
\end{proof}

\begin{lemma} \label{prop:v-to-path}
Let $G$ be a graph with $\delta(G)=\delta$ and $x,y,w$ be arbitrary vertices of $G$. Then, for every shortest path $[x,y]$ connecting $x$ with $y$,
$d(w,[x, y]) \leq (x|y)_w + 2 \delta + \frac{1}{2}$ holds.
\end{lemma}
	
\begin{proof} Consider in $G$ a geodesic triangle $\Delta(x,y,w)$ formed by $[x,y]$ and two arbitrary shortest paths $[x,w]$ and $[y,w]$.
Let $c$ be a vertex on $[x,y]$ at distance $\lfloor(y|w)_x\rfloor$ from $x$. We have $(x|y)_w\geq \min\{(x|c)_w,(y|c)_w\} - \delta$.

If  $(x|c)_w\leq (y|c)_w$, then $(x|c)_w -(x|y)_w\leq \delta$.  Therefore, $(x|w)_c=d(x,c)-(c|w)_x\leq (y|w)_x -(c|w)_x= (x|c)_w-(x|y)_w\leq \delta$.
As $d(w,c)=(x|c)_w+(x|w)_c\leq (x|y)_w + \delta +\delta$, we get $d(w,[x,y])\le d(w,c)\leq (x|y)_w + 2\delta$.

If $(x|c)_w\geq (y|c)_w$, then $(y|c)_w -(x|y)_w\leq \delta$. Therefore, $(y|w)_c=d(y,c)-(c|w)_y\leq (x|w)_y +\frac{1}{2}-(c|w)_y= (y|c)_w-(x|y)_w+\frac{1}{2}\leq \delta+\frac{1}{2}$.
As $d(w,c)=(y|c)_w+(y|w)_c\leq (x|y)_w + \delta +\delta+\frac{1}{2}$, we get $d(w,[x,y])\le d(w,c)\leq (x|y)_w + 2\delta+\frac{1}{2}$.
\end{proof}

\begin{lemma} \label{lm:hb-vs-ins} $\tau(G)=\iota(G)\leq 4 \delta(G)$ and $\varsigma(G)\le 3\delta(G)+\frac{1}{2}$.
\end{lemma}

\begin{proof}  Let $\delta:=\delta(G)$. Pick a geodesic triangle $\Delta(x,y,z)$ of $G$ formed by shortest paths
$[x,y]$, $[x,z]$, and $[y,z]$. Pick also the vertices $y_z\in [x,y]$ and $z_y\in [x,z].$
Evidently, $(y_z|y)_x=d(x,y_z)=\lfloor(y|z)_x\rfloor=d(x,z_y)=(z_y|z)_x$.
We also have $(y_z|z_y)_x \ge \min \{(y_z|y)_x,(y|z_y)_x\}-\delta \ge \min \{(y_z|y)_x,(y|z)_x, (z|z_y)_x\}-2\delta$.
It implies that $(y_z|z_y)_x \ge \lfloor (y|z)_x\rfloor -2\delta$. Consequently, $d(x,y_z)+d(x,z_y)-d(y_z,z_y)\ge 2 \lfloor (y|z)_x \rfloor-4\delta$ holds,
implying $d(y_z,z_y)\le 4\delta$.

To prove $\varsigma(G)\le 3\delta+\frac{1}{2}$, consider a geodesic triangle
$\Delta(x,y,z)$ formed by shortest paths $[x,y], [x,z]$, and $[y,z]$ and let $w$ be an arbitrary vertex from $[x,y]$.
Without loss of generality, suppose that $(x|z)_w \leq (y|z)_w$. Since $w$ is on a shortest path between $x$ and $y$,
we have $0=(x|y)_w \geq \min\{(x|z)_w,(y|z)_w\} - \delta=(x|z)_w - \delta$,  {i.e.},
$(x|z)_w \le\delta.$ 	By Lemma~\ref{prop:v-to-path}, $d(w,[x, z]) \leq (x|z)_w + 2 \delta + \frac{1}{2}\le 3\delta + \frac{1}{2}.$	
\end{proof}

\begin{lemma}
  $\delta(G) \leq \tau(G) + \frac{1}{2}$.
\end{lemma}

\begin{proof}
  Let $\tau:=\tau(G)$. Consider four vertices $w,x,y,z$ and assume
  without loss of generality that $d(w,y) +d(x,z) \geq \max
  \{d(w,x)+d(y,z),d(w,z)+d(x,y)\}$. Pick a geodesic triangle
  $\Delta(w,x,y)$ of $G$ formed by three arbitrary shortest paths
  $[w,x]$, $[w,y]$, and $[x,y]$. Pick a geodesic triangle
  $\Delta(w,y,z)$ of $G$ formed by the shortest path $[w,y]$ and two
  arbitrary shortest paths $[w,z],[y,z]$.

  Without loss of generality, assume  that $(x|y)_w \leq (y|z)_w$. Let $x_y$
  and $y_x$ be respectively the vertices of $[w,x]$ and $[w,y]$ at
  distance $\lfloor(x|y)_w\rfloor$ from $w$. Let $z'$ be the vertex of
  $[w,z]$ at distance $\lfloor(x|y)_w\rfloor \leq
  \lfloor(y|z)_w\rfloor$ from $w$.  Since
  $d(x_y,y_x) \leq \tau$ and $d(y_x,z') \leq \tau$, by the triangle
  inequality, we have:
  \begin{align*}
    d(w,y)+d(x,z) &\le (d(w,y_x)+d(y_x,y))+(d(x,x_y)+ 2 \tau +d(z',z))\\
    &= d(w,y_x)+d(z',z) + d(y,y_x) + d(x,x_y) + 2\tau\\
    &\leq d(w,z')+d(z',z) + d(x,y) + 1 + 2\tau\\
    &= d(w,z) + d(x,y) + 2\tau + 1.
  \end{align*}

  This establishes the four-point condition for $w,x,y,z$, and
  consequently $\delta(G) \leq \tau + \frac{1}{2}$.
\end{proof}

\begin{lemma}\label{lm:interval_thin}
  $\kappa(G)\leq \min\{\tau(G), 2\delta(G), 2 \varsigma(G)\}$.
\end{lemma}

\begin{proof} Let $u,v$ be two arbitrary vertices of $G$ and let
  $x,y\in I(u,v)$ such that $d(u,x)=d(u,y)$. Since $d(u,x)+d(y,v) =
  d(u,y)+d(x,v)=d(u,v)$, we have $d(u,v)+d(x,y) \leq d(u,v) + 2
  \delta(G)$ and consequently, $d(x,y) \leq 2\delta(G)$. Thus
  $\kappa(G) \leq 2 \delta(G)$.  Let $[u,v]$ be any shortest
  $(u,v)$-path passing through $y$ and $[u,x],[x,v]$ be two arbitrary
  shortest $(u,x)$- and $(x,v)$-paths.  Consider the geodesic triangle
  $\Delta(x,u,v):=[u,x]\cup [x,v]\cup [v,u]$.  We have
  $(x|v)_u=(d(x,u)+d(u,v)-d(x,v))/2=d(x,u)=d(y,u)$. Hence, if
  $\Delta(x,u,v)$ is $\tau$-thin, then $d(x,y)\leq \tau$.  That is,
  $\kappa(G)\leq \tau(G)$. If $\Delta(x,u,v)$ is $\varsigma$-slim,
  then there is a vertex $z\in [u,x]\cup [x,v]=[u,v]$ such that
  $d(y,z)\leq \varsigma$. Necessarily, $d(x,z)\leq \varsigma$ as well,
  implying $d(x,y)\leq 2\varsigma$. Thus, $\kappa(G)\leq 2
  \varsigma(G)$.
\end{proof}

\begin{remark} In general, the converse of the inequality $\kappa(G)\leq 2 \delta(G)$ from Proposition \ref{prop:sl-vs-th}
does not hold: for odd cycles $C_{2k+1}$, $\kappa(C_{2k+1})=0$ while $\delta(C_{2k+1})$ increases with $k$.
However, the following result holds.  If $G$ is a graph, denote by $G'$ the graph obtained by subdividing all
edges of $G$ once.  Papasoglu \cite{Pa1995}  showed that if $G'$ has $\kappa$-thin intervals,
then $G$ is $f(\kappa)$-hyperbolic  for  some  function $f$ (which may be exponential).
\end{remark}

\section{Geodesic spanning trees}\label{GS-trees}
In this section, we prove that any complete geodesic metric space $(X,d)$ has a geodesic spanning tree rooted at any basepoint $w$. We hope that this general result will be
useful in other contexts. For finite graphs this is well-known and simple, and such trees can be constructed in various ways, for example via Breadth-First-Search. The existence
of BFS-trees in infinite graphs has been established by Polat \cite[Lemma 3.6]{Polat}. However for geodesic spaces this result seems to be new (and not completely trivial) and we
consider it as one of the main results of the paper.   A \emph{geodesic spanning tree rooted at a point} $w$ (a \emph{GS-tree} for short) of a geodesic space
$(X,d)$ is a union of geodesics  $\Gamma_w:=\bigcup_{x\in X}\gamma_{w,x}$ with one
end at $w$ such that $y\in \gamma_{w,x}$ implies that $\gamma_{w,y}\subseteq \gamma_{w,x}$. Then $X$ is the union of the images $[w,x]$ of the geodesics of $\gamma_{w,x}\in \Gamma_w$
and one can show that
there exists a real tree $T=(X,d_T)$ such that any $\gamma_{w,x}\in \Gamma_w$ is the $(w,x)$-geodesic of $T$. Finally recall that
a metric space $(X,d)$ is called \emph{complete} if every Cauchy sequence of  $X$  has a limit in $X$.

\begin{theorem} \label{BFS_geodesic} For any complete geodesic metric space $(X,d)$ and for any basepoint $w$
one can define a geodesic spanning tree  $\Gamma_w=\bigcup_{x\in X}\gamma_{w,x}$
rooted at $w$ and a real tree $T=(X,d_T)$ such that any $\gamma_{w,x}\in \Gamma_w$ is the unique $(w,x)$-geodesic of $T$.
\end{theorem}

The existence of a geodesic spanning tree   $\Gamma_w=\bigcup_{x\in X}\gamma_{w,x}$
rooted at $w$  follows from the following proposition:

\begin{proposition} \label{prop_geodesic}
  For any complete geodesic metric space $(X,d)$, for any pair of
  points $x,y\in X$ one can define an $(x,y)$-geodesic $\gamma_{x,y}$
  such that for all $x,y \in X$ and for all $u,v \in \gamma_{x,y}$, we
  have $\gamma_{u,v} \subseteq \gamma_{x,y}$.
\end{proposition}

\begin{proof}
 Let $\preceq$ be a well-order on $X$. For any $x, y \in X$ we define
 inductively  two sets $P^{\prec v}_{x,y}$ and
 $P^v_{x,y}$ for any $v \in X$:
 \begin{align*}
   P^{\prec v}_{x,y}  &=  \{x,y\} \cup \bigcup_{u\prec v} P^u_{x,y},\\
   P^v_{x,y}  &=
    \begin{cases}
      P^{\prec v}_{x,y}\cup \{v\}& \text{if there is an
        $(x,y)$-geodesic $\gamma$ with $P^{\prec v}_{x,y}\cup
        \{v\} \subseteq \gamma$},\\
     P^{\prec v}_{x,y} & \text{otherwise.}\\
    \end{cases}
  \end{align*}

 We set $P_{x,y} := \bigcup_{u \in X} P^u_{x,y}$.

 \begin{claim}\label{claim-induction-geodesics}
   For all $x,y \in X$ and for any $v \in X$,
   \begin{enumerate}
   \item there exists an $(x,y)$-geodesic $\gamma^{\prec v}_{x,y}$
     such that $P^{\prec v}_{x,y} \subseteq \gamma^{\prec v}_{x,y}$,
   \item there exists an $(x,y)$-geodesic $\gamma^{v}_{x,y}$
     such that $P^{ v}_{x,y} \subseteq \gamma^{ v}_{x,y}$,
   \item there exists an $(x,y)$-geodesic $\gamma_{x,y}$ such that
     $P_{x,y} \subseteq \gamma_{x,y}$.
   \end{enumerate}
 \end{claim}

 \begin{proof}
   We prove the claim by transfinite induction on the well-order $\preceq$.

   \medskip
   \textit{To (1):} Assume that for any $u \prec v$, there exists an
     $(x,y)$-geodesic $\gamma^u_{x,y}$ such that $P^{u}_{x,y}
     \subseteq \gamma^u_{x,y}$. If $ P^{\prec v}_{x,y} = \{x,y\}$
     (this happens in particular if $v$ is the least element of $X$
     for $\preceq$), then let $\gamma^{\prec v}_{x,y}$ be any
     $(x,y)$-geodesic. If there exists $u \prec v$ such that
     $P^u_{x,y} = P^{\prec v}_{x,y}$, then let $\gamma^{\prec v}_{x,y}
     = \gamma^u_{x,y}$.

     Suppose now that $P^{\prec v}_{x,y} \neq \{x,y\}$ and that for
     any $u \prec v$, $P^u_{x,y} \subsetneq P^{\prec v}_{x,y}$. Note
     that for $u \in P^{\prec v}_{x,y} \setminus \{x,y\}$, we have $u
     \in P^u_{x,y}$, and for any $u \preceq w \prec v$,
     $\gamma_{x,y}^w(d(x,u))= u$.

     Let $D := \{t \in [0,d(x,y)]: \forall \varepsilon > 0, \exists
     u \in P_{x,y}^{\prec v} \mbox{ such that } |d(x,u) - t| \leq \epsilon\}$.
     Note that
     $D$ is a closed subset of $[0,d(x,y)]$ and that for any $u \in
     P_{x,y}^{\prec v}$, $d(x,u) \in D$.  We define $\gamma =
     \gamma_{x,y}^{\prec v}$ in two steps: we first define $\gamma$ on
     $D$ and then we extend it to the whole segment $[0,d(x,y)]$.

     For any $t \in D$, there exists a sequence $(u_i)_{i \in \N}$
     such that for every $i$, $u_i \in P_{x,y}^{\prec v}$, $|d(u_i,x)
     - t| \leq 1/i$. Set $t_i:=d(x,u_i)$. For any $i < j \in \N$, let $u^* := \max_\prec
     (u_i,u_j)$ and note that $d(u_i,u_j) =
     d(\gamma_{x,y}^{u^*}(t_i), \gamma_{x,y}^{u^*}(t_j)) = |
     t_i - t_j| \leq |t_i - t| + |t- t_j| \leq
     1/i+1/j \leq 1/2i$. Consequently, $(u_i)_{i \in \N}$ is a Cauchy
     sequence in $(X,d)$ and thus $(u_i)_{i \in \N}$ converges to a
     point $u \in X$ since $(X,d)$ is complete. Note that $u$ is
     independent of the choice of the sequence $(u_i)_{i \in \N}$, and
     let $\gamma(t) = u$. For any $u \in P_{x,y}^{\prec v}$, $d(x,u)
     \in D$ and it is easy to see that $\gamma(d(x,u)) = u$ ({i.e.},
     $\gamma$ contains $P_{x,y}^{\prec v}$).  Moreover, note that by triangle inequality
     $|d(u_i,u) - d(u_j,u)| \leq d(u_i,u_j) \leq 1/i+ 1/j$ for any
     $i,j$, and consequently, $d(u_i,u) \leq 1/i$.

     For any $t,t' \in D$, we claim that
     $d(\gamma(t),\gamma(t')) = |t - t'|$. Consider two sequences
     $(u_i)_{i \in \N}$ and $(u'_i)_{i \in \N}$ such that for every
     $i$, $|d(u_i,x) - t| \leq 1/i$ and $|d(u_i',x) - t'| \leq
     1/i$. Set $t_i:=d(x,u_i)$ and $t'_i:=d(x,u'_i)$.  Consider the
     respective limits $u = \gamma(t)$ and $u'=\gamma(t')$ of
     $(u_i)_{i \in \N}$ and $(u'_i)_{i \in \N}$.  For every $i$, let
     $u^* = \max_\prec (u_i,u_i')$ and note that
     $d(u_i,u_i') = d(\gamma_{x,y}^{u^*}(t_i),
     \gamma_{x,y}^{u^*}(t'_i)) = | t_i - t'_i|$. By the continuity of
     the distance function $d(\cdot,\cdot)$, we thus have
     $d(u,u') = d(\gamma(t),\gamma(t')) = |t-t'|$.

       Suppose now that $\gamma$ is defined on $D$. For every interval
       $[t_0,t_1] \subseteq [0,d(x,y)]$ such that $[t_0,t_1]\cap D =
       \{t_0,t_1\}$, let $\gamma_{t_0,t_1}: [0,t_1 - t_0] \to X$ be an
       arbitrary $(\gamma(t_0),\gamma(t_1))$-geodesic (it exists since
       $d(\gamma(t_0),\gamma(t_1)) = t_1-t_0$ and $(X,d)$ is
       geodesic). For any $t \in [t_0,t_1]$, let $\gamma(t) =
       \gamma_{t_0,t_1}(t-t_0)$.

       For any $0 \leq t < t'\leq d(x,y)$, we claim that
       $d(\gamma(t),\gamma(t')) \leq t'-t$. Let $t_0 :=
       \sup(D\cap[0,t])$, $t_1 := \inf(D\cap[t,d(x,y)])$, $t_0' :=
       \sup(D\cap[0,t'])$, $t_1' := \inf(D\cap[t',d(x,y)])$. If $t_0' <
       t_1$, then $t_0 = t_0' \leq t < t' \leq t_1' = t_1$ and
       $d(\gamma(t),\gamma(t')) =
       d(\gamma_{t_0,t_1}(t-t_0),\gamma_{t_0,t_1}(t'-t_0)) = t'-t$.
       Otherwise, we have $t_0 \leq t \leq t_1 \leq t_0' \leq t' \leq
       t_1'$. If $t = t_1$, then $d(\gamma(t),\gamma(t_1)) = t_1 -
       t=0$. Otherwise, since $t \in [t_0,t_1]$ and $[t_0,t_1] \cap D =
       \{t_0,t_1\}$, $d(\gamma(t),\gamma(t_1)) =
       d(\gamma_{t_0,t_1}(t-t_0),\gamma_{t_0,t_1}(t_1-t_0)) =
       t_1-t$. Similarly, $d(\gamma(t_0'),\gamma(t')) = t' -
       t_0'$. Since $t_1, t_0' \in D$, we already know that
       $d(\gamma(t_1),\gamma(t_0')) = t_0' - t_1$. Consequently, $$t' -
       t = t' - t_0' + t_0' - t_1 + t_1 - t =
       d(\gamma(t'),\gamma(t_0')) + d(\gamma(t_0'),\gamma(t_1))+
       d(\gamma(t_1),\gamma(t)) \geq d(\gamma(t'),\gamma(t)).$$

       Suppose now that there exists $0 \leq t < t' \leq d(x,y)$ such
       that $d(\gamma(t),\gamma(t')) < t'-t$. Then $d(x,y) \leq
       d(\gamma(0),\gamma(t)) + d(\gamma(t),\gamma(t')) +
       d(\gamma(t'),\gamma(d(x,y))) < t - 0 + t'-t+d(x,y) - t' =
       d(x,y)$, a contradiction. Consequently, for any $0 \leq t < t'
       \leq d(x,y)$, we have $d(\gamma(t),\gamma(t')) = t'-t$ and thus
       $\gamma$ is an $(x,y)$-geodesic containing $P_{x,y}^{\prec v}$.

\medskip
   \textit{To (2):} If $P^v_{x,y} = P^{\prec v}_{x,y}$, the property holds by the
     previous statement of the claim. Otherwise, $P^v_{x,y} =
     P^{\prec v}_{x,y} \cup \{v\}$, and the property holds by the
     definition of $P^v_{x,y}$.

\medskip
\textit{To (3):}  If there exists $v \in X$ such that $X$ coincides with
$\{u\in X: u \preceq v\}$, then we are done by the previous
     statement of the claim. Otherwise, the proof is identical to the proof
     of statement (1) of the claim.
 \end{proof}

 \begin{claim}\label{claim-Pgeo}
   $P_{x,y}$ is an $(x,y)$-geodesic.
 \end{claim}

 \begin{proof}
   By Claim~\ref{claim-induction-geodesics}, there exists an
   $(x,y)$-geodesic $\gamma_{x,y}$ such that $P_{x,y} \subseteq
   \gamma_{x,y}$.  Conversely, for any $v \in \gamma_{x,y}$, since
   $P^{\prec v}_{x,y} \subseteq P_{x,y}$, $\gamma_{x,y}$ is an
   $(x,y)$-geodesic containing $P^{\prec v}_{x,y} \cup
   \{v\}$. Therefore, by the definition of $P^v_{x,y}$, $v \in P^v_{x,y}
   \subseteq P_{x,y}$.
 \end{proof}

Let $B(x,r)$ denotes the {\it closed ball} of radius $r$ centered at a point $x$ of $(X,d)$.

 \begin{claim}\label{claim-inclusion}
   For all $x, y \in X$ and for any $u \in P_{x,y}$, $P_{x,u} = P_{x,y}
   \cap B(x,d(x,u))$.
 \end{claim}

 \begin{proof}
   Let $\gamma_1 := P_{x,y} \cap B(x,d(u,x))$ and $\gamma_2 := P_{x,y}
   \cap B(y,d(u,y))$. Note that $P_{x,y} = \gamma_1 \cup \gamma_2$,
   that $\gamma_1$ is an $(x,u)$-geodesic, and that $\gamma_2$ is a
   $(u,y)$-geodesic. Let $\gamma_3 := P_{x,u} \cup \gamma_2$, and note
   that $\gamma_3$ is an $(x,y)$-geodesic.

   We prove the claim by induction on $\preceq$. Note that if for any
   $w \prec v$, $P^w_{x,u} = P^w_{x,y} \cap B(x,d(x,u))$, then $P^{\prec
     v}_{x,u} = \bigcup_{w \prec v} P^w_{x,u} = \bigcup_{w \prec v}
   \left(P^w_{x,y} \cap B(x,d(x,u)) \right) = P^{\prec v}_{x,y} \cap
   B(x,d(x,u))$.  If $v \in P^v_{x,y}\cap B(x,d(x,u)) \subseteq
   \gamma_1$, then $\gamma_1$ is an $(x,u)$-geodesic containing
   $P^{\prec v}_{x,u} \cup \{v\}$, and by the definition of $P^v_{x,u}$,
   we have $v \in P^v_{x,u} \subseteq P_{x,u}$.  Conversely, suppose
   that $v \in P^v_{x,u} \subseteq P_{x,u} \subseteq P_{x,u}\cup \gamma_2
   = \gamma_3$. Since $\{v\} \cup (P^{\prec v}_{x,y} \cap B(x, d(u,x)))
   \subseteq P_{x,u}$ and $P^{\prec v}_{x,y}\cap B(y,d(u,y)) \subseteq
   \gamma_2$, $\gamma_3$ is an $(x,y)$-geodesic containing $P^{\prec
     v}_{x,y} \cup \{v\}$. By the definition of $P^v_{x,y}$, we have $v
   \in P^v_{x,y} \subseteq P_{x,y}$
 \end{proof}

 By Claim~\ref{claim-Pgeo}, we can consider the set of geodesics
 $\{P_{x,y}:  x,y \in X\}$.  For all $x, y \in X$ and for any
 $u,v\in P_{x,y}$ such that $d(v,x) < d(u,x)$, by
 Claim~\ref{claim-inclusion}, $P_{u,v} \subseteq P_{x,u} \subseteq
 P_{x,y}$. This finishes the proof of Proposition \ref{prop_geodesic}.
\end{proof}

Consequently, $\Gamma_w=\bigcup_{x\in X}\gamma_{w,x}$ is a geodesic spanning tree of $(X,d)$ rooted at $w$.
For any $x\in X$, denote by $[w,x]$ the geodesic segment between $x$ and $w$ which is the image of the geodesic $\gamma_{x,w}$.
From the definition of $\Gamma_w$, if $x'\in [x,w]$, then $[x',w]\subseteq [x,w]$.
From the continuity of geodesic maps
and the definition of $\Gamma_w$ it follows that for any two geodesics $\gamma_{w,x},\gamma_{w,y}\in \Gamma_w$
the intersection $[w,x]\cap [w,y]$ is the image $[z,w]$ of some geodesic $\gamma_{w,z}\in \Gamma_w$. Call $z$ the \emph{lowest common ancestor}
of $x$ and $y$ (with respect to the root $w$) and denote it by $\lca(x,y)$. Define $d_T$ by setting $d_T(w,x):=d(w,x)$ and
$d_T(x,y):=d(w,x)+d(w,y)-2d(w,z)=d(x,z)+d(z,y)$ for any
two points $x,y\in X$.

The existence of a real tree $T=(X,d_T)$ such that any
$\gamma_{w,x}\in \Gamma_w$ is the unique $(w,x)$-geodesic of $T$
immediately follows from the following proposition:

\begin{proposition} \label{realtree} $T=(X,d_T)$ is a real tree and any $\gamma_{w,x}\in \Gamma_w$ is the unique $(w,x)$-geodesic of $T$.
\end{proposition}

\begin{proof} From the definition, $d_T(w,x)=d(w,x)$ and $d_T(x,y)\ge d(x,y)$ for any $x,y\in X$.
For a pair of points $x,y\in X$, set $z:=\lca(x,y)$. Denote by $[x,z]$ the portion of the geodesic segment $[x,w]$
between $x$ and $z$ and by $[y,z]$ the portion of the geodesic segment $[y,w]$ between $y$ and $z$. Then $[x,z]$ and $[y,z]$
are geodesic segments of $(X,d)$, and thus they are geodesic segments of $T$. Let $[x,y]:=[x,z]\cup [z,y]$.
We assert that $[x,y]$ is a geodesic segment of $T$.
Suppose that $[x,z]$ and $[z,y]$ are the images of the geodesics $\gamma_{T,x,z}$ and $\gamma_{T,y,z}$ of $(X,d)$, respectively.
Let $\gamma_{T,x,y}$ denotes the continuous map from $[0,d_T(x,y)]$ to $X$ such that $\gamma_{T,x,y}(t)=\gamma_{T,x,z}(t)$ if $0\le t\le d(x,z)$
and $\gamma_{T,x,y}(t)=\gamma_{T,z,y}(t-d(x,z))$ if $d(x,z)\le t\le d(x,z)+d(z,y)$. Clearly, $[x,y]$ is the image of $\gamma_{T,x,y}$ and  $z=\gamma_{T,x,y}(d(x,z))$. Let
$0\le t<t'\le d(x,z)+d(z,y)$ and let $u:=\gamma_{T,x,y}(t)$ and $v:=\gamma_{T,x,y}(t')$. If $t,t'\le d(x,z)$, then $u,v\in [x,z]$ and
one can easily see that $d_T(u,v)=d(u,v)=t'-t$. Analogously if $t,t'\ge d(x,z)$, then $d_T(u,v)=d(u,v)=t'-t$. Now, let
$t\le d(x,z)\le t'$. Then one can easily see that $\lca(u,v)=\lca(x,y)=z$. Consequently, $d_T(u,v)=d(u,z)+d(z,v)=(d(x,z)-t)+(t'-d(x,z))=t'-t$ and therefore $[x,y]$
is a geodesic segment of $T$ and $\gamma_{T,x,y}$ is a geodesic map.

Let $x,y,u$ be any triplet of points of $X$ and set $z:=\lca(x,y), z':=\lca(x,u)$, and $z'':=\lca(u,y)$. Suppose without loss of generality that $d(u,z')\le d(u,z'')$.
Since $z',z''$ belong to $[u,w]$ and $[z',w]\cup [z'',w]\subseteq [u,w]$, necessarily $z''\in [z',w]$. Since $z'\in [x,w]$, we conclude that $z''\in [x,w]$.
Since we also have $z''\in [y,w]$, from the definition of $z$ we deduce that
$z''\in [z,w]$. If $z\ne z''$, from the definition of $z''$ we conclude that $z\notin [z',z'']$, {i.e.}, $z\in [x,z']$. In this case, $z'\in [z,w]\subseteq [y,w]$, yielding $z''=z'$.
This show that either (1) $z=z''\in [z',w]$ or (2) $z'=z''\in [z,w]$. We will use this conclusion to prove that $T$ is a real tree.

First we show that $T$ is uniquely geodesic, i.e., that for any points
$x,y,u$ such that $d_T(x,y)=d_T(x,u)+d_T(u,y)$, $u$ belongs to
$[x,y]$.  Since $z''\in [z',w]$,
$d_T(x,u)+d_T(u,y)=d(x,z')+2d(z',u)+d(z',z'')+d(z'',y).$ Since
$d_T(x,y)=d(x,z)+d(z,y)$ and $d_T(x,y)=d_T(x,u)+d_T(u,y)$, we obtain
that $d(x,z)+d(z,y)=d(x,z')+2d(z',u)+d(z',z'')+d(z'',y)$. If
$z'=z''\in [z,w]$ this equality is possible only if $z=z'=z''$ and
$d(z',u)=0$. Therefore, in this case $u=z'=z\in [x,y]$. If
$z=z''\in [z',w]$, then again the previous equality is possible only
if $u=z' \in [x,z]\subseteq [x,y]$. Thus $[x,y]$ is the unique
geodesic segment connecting $x$ and $y$ in $T$.

Now suppose that $[x,u]\cap [u,y]=\{ u\}$ and we assert that
$[x,u]\cup[u,y]=[x,y]$. Obviously, it suffices to show that $u\in
[x,y]$. Note that by the definitions of $z'$ and $z''$ and since $z'
\in [u,z'']$, we have $[u,z'] \subseteq [u,x]\cap[u,y]$.  Since
$[x,u]\cap [u,y]=\{ u\}$, necessarily $u = z'$. Observe also that if
$z' \notin [x,z]$, then $z \neq z'$, $z' = z''$, $z \in [x,z']$, and
$[z,u] = [z,z'] = [z,z''] \in [x,z'] \cap [y,z''] = [x,u] \cap [y,u]$,
a contradiction. Consequently $u = z' \in [x,z] \subseteq [x,y]$. This
finishes the proofs of Proposition \ref{realtree} and Theorem
\ref{BFS_geodesic}.
\end{proof}

\begin{remark}\label{rem:GS-not-BFS}
  The proof of Theorem~\ref{BFS_geodesic} of the  existence of 
  GS-trees is completely different from the proof of
  Polat~\cite{Polat} of the existence of BFS-trees in arbitrary
  graphs. The proof of \cite{Polat}, as the usual BFS-tree construction in
  finite graphs, constructs an increasing sequence of trees that span
  vertices at larger and larger distances from the root. In other
  words, from an arbitrary well-ordering of the set $V$ of vertices of
  $G$, Polat \cite{Polat} constructs a well-ordering of $V$ that 
  is consistent with the distances to the root.

  When considering arbitrary geodesic metric spaces, a well-ordering
  consistent with the distances to the basepoint $w$ does not always
  exist; consider for example the segment $[0,1]$ with $w=0$.
\end{remark}

\section{Fast approximation} \label{approx}  

In this section, we introduce a new parameter of a graph $G$ (or of a
geodesic space $X$), the rooted insize. This parameter depends on an
arbitrary fixed BFS-tree of $G$ (or a GS-tree of $X$). It can be
computed efficiently and it provides constant-factor approximations
for $\delta(G)$, $\varsigma(G)$, and $\tau(G)$. In particular, we
obtain a very simple factor 8 approximation algorithm (with an
additive constant 1) for the hyperbolicity $\delta(G)$ of an
$n$-vertex graph $G$ running in optimal $O(n^2)$ time (assuming that
the input is the distance matrix of $G$).\footnote{In all algorithmic
  results, we assume the word-RAM model.}

\subsection{Fast approximation of hyperbolicity}

Consider a graph $G=(V,E)$ and an arbitrary BFS-tree $T$ of $G$ rooted
at some vertex $w$.  Denote by $x_y$ the vertex of $[w,x]_T$ at
distance $\lfloor(x|y)_w\rfloor$ from $w$ and by $y_x$ the vertex of
$[w,y]_T$ at distance $\lfloor(x|y)_w\rfloor$ from $w$.  Let
$\rho_{w,T}(G):=\max \{ d(x_y,y_x): x,y\in V\}.$  In some sense,
$\rho_{w,T}(G)$ can be seen as the insize of $G$ with respect to $w$
and $T$. For this reason, we call $\rho_{w,T}(G)$ the \emph{rooted
  insize} of $G$ with respect to $w$ and $T$. The differences between
$\rho_{w,T}(G)$ and $\iota(G)$ are that we consider only geodesic
triangles $\Delta(w,x,y)$ containing $w$ where the geodesics $[w,x]$
and $[w,y]$ belong to $T$, and we consider only $d(x_y,y_x)$, instead
of $\max \{d(x_y,y_x), d(x_w,w_x), d(y_w,w_y)\}$.  Using $T$, we can
also define the \emph{rooted thinness} of $G$ with respect to $w$ and
$T$: let
$\mu_{w,T}(G) = \max\big\{d(x',y') : \exists x,y \in V \text{ such that } x'
\in [w,x]_T, y' \in [w,y]_T \text{ and } d(w,x') = d(w,y') \leq
(x|y)_w\big\}$.

Similarly, for a geodesic space $(X,d)$ and an arbitrary GS-tree $T$
rooted at some point $w$ (see Section \ref{GS-trees}), denote by $x_y$
the point of $[w,x]_T$ at distance $(x|y)_w$ from $w$ and by $y_x$ the
point of $[w,y]_T$ at distance $(x|y)_w$ from $w$. Analogously, we
define the \emph{rooted insize} of $(X,d)$ with respect to $w$ and $T$
as $\rho_{w,T}(X):=\sup \{ d(x_y,y_x): x,y\in X\}$. We also define the
\emph{rooted thinness} of $(X,d)$ with respect to $w$ and $T$ as
$\mu_{w,T}(X) = \sup\big\{d(x',y') : \exists x,y \in X \text{ such
  that } x' \in [w,x]_T, y' \in [w,y]_T \text{ and } d(w,x') = d(w,y')
\leq (x|y)_w\big\}$.

Using the same ideas as in the proofs of
Propositions~\ref{hyp_charact} and~\ref{prop:sl-vs-th} establishing
that $\iota(X) = \tau(X)$ and $\iota(G) = \tau(G)$, we can show that
these two definitions give rise to the same value.

\begin{proposition}\label{prop:rhomu}
  For any geodesic space $X$ and any GS-tree $T$ rooted at a point
  $w$, $\rho_{w,T}(X) = \mu_{w,T}(X)$.  Analogously, for any graph
  $G$ and any BFS-tree $T$ rooted at $w$, $\rho_{w,T}(G) =
  \mu_{w,T}(G)$.
\end{proposition}

In the following, when $G$ (or $X$), $w$ and $T$ are clear from the context, we
denote $\rho_{w,T}(G)$ (or $\rho_{w,T}(X)$) by $\rho$.
The next theorem is the main result of this paper. It establishes that
$2\rho$ provides an 8-approximation of the hyperbolicity of
$\delta(G)$ or $\delta(X)$, and that in the case of a finite graph
$G$, $\rho$ can be computed in $O(n^2)$ time when the distance matrix
$D$ of $G$ is given.

\begin{theorem} \label{main}
  Given a graph $G$ (respectively, a geodesic space $X$) and a
  BFS-tree $T$ (respectively, a GS-tree $T$) rooted at $w$,
  \begin{enumerate}[(1)]
  \item $\delta(G)\leq 2\rho_{w,T}(G) +1\le 8\delta(G) +1$ (respectively,
    $\delta(X)\leq 2\rho_{w,T}(X)\le 8\delta(X)$).
  \item If $G$ has $n$ vertices, given the distance matrix $D$ of $G$,
    the rooted insize $\rho_{w,T}(G)$ can be computed in $O(n^2)$
    time.  Consequently, an 8-approximation (with an additive constant
    1) of the hyperbolicity $\delta(G)$ of $G$ can be found in
    $O(n^2)$ time.
  \end{enumerate}
\end{theorem}

\begin{proof}
  We prove the first assertion of the theorem for graphs (for geodesic
  spaces, the proof is similar). Let $\rho := \rho_{w,T}(G)$,
  $\delta:=\delta(G)$, and $\delta_w:=\delta_w(G)$. By Gromov's
  Proposition \ref{hyp_basepoint}, $\delta\le 2\delta_w$. We proceed
  in two steps. In the first step, we show that $\rho\le 4\delta$. In
  the second step, we prove that $\delta_w\le \rho+\frac{1}{2}$. Hence,
  combining both steps we obtain $\delta\le 2\delta_w\le 2 \rho+1\le
  8\delta+1$.

  The first step follows from Proposition~\ref{prop:sl-vs-th} and from
  the inequality $\rho\le \iota(G)=\tau(G)$.  To prove that
  $\delta_w \le \rho+1/2$, for any quadruplet $x,y,z,w$ containing
  $w$, we show the four-point condition
  $d(x,z)+d(y,w)\le \max\{
  d(x,y)+d(z,w),d(y,z)+d(x,w)\}+(2\rho+1)$. Assume without loss of
  generality that
  $d(x,z)+d(y,w) \ge \max\{d(x,y)+d(z,w),d(y,z)+d(x,w)\}$ and that
  $d(w,x_y) = d(w,y_x) \le d(w,y_z) = d(w,z_y)$. Since $y_x,y_z$
  belong to the shortest path $[w,y]$ of $T$ (that is also a shortest
  path of $G$), we have  $d(y_x,y_z) = d(y,y_x) - d(y,y_z)$.
  From the definition of $\rho$, we also have $d(x_y,y_x)\le \rho$ and
  $d(y_z,z_y)\le \rho$.  Consequently, by the definition of
  $x_y, y_x, y_z, z_y$ and by the triangle inequality, we get
  \begin{align*}
    d(y,w)+d(x,z)& \le
    d(y,w)+d(x,x_y)+d(x_y,y_x)+d(y_x,y_z)+d(y_z,z_y)+d(z_y,z)\\
    &\le (d(y,y_z)+d(y_z,w))+d(x,x_y)+\rho+
    d(y_x,y_z)+\rho+d(z_y,z)\\
    &= d(y,y_z)+d(w, z_y) + d(x,x_y)+  d(y_x,y_z)+ d(z_y,z)+ 2 \rho\\
    &= d(y,y_z)+d(x,x_y)+ (d(y,y_x)-d(y,y_z))+(d(w, z_y) + d(z_y,z))+2 \rho \\
    &= d(y,y_z)+d(x,x_y)+ d(y,y_x)-d(y,y_z)+d(w,z)+2 \rho \\
    & \le d(x,y)+1 + d(w,z) + 2 \rho,
  \end{align*}
  the last inequality following from the definition of $x_y$ and $y_x$
  in graphs (in the case of geodesic metric spaces, we have
  $d(x,x_y)+d(y,y_x) = d(x,y)$). This establishes the four-point
  condition for $w,x,y,z$ and proves that $\delta_w\le \rho+1/2$.

  We present now a simple self-contained algorithm for computing the
  rooted insize $\rho$ in $O(n^2)$ time when $G = (V,E)$ is a graph
  with $n$ vertices.
  For any non-negative integer $r$, let $x(r)$ be the unique vertex of
  $[w,x]_T$ at distance $r$ from $w$ if $r<d(w,x)$ and the vertex $x$
  if $r\ge d(w,x)$.  First, we compute in $O(n^2)$ time a table $M$
  with rows indexed by $V$, columns indexed by $\{1,\hdots, n\}$, and
  such that $M(x,r)$ is the identifier of the vertex $x(r)$ of
  $[w,x]_T$ located at distance $r$ from $w$.  To compute this table,
  we explore the tree $T$ starting from $w.$ Let $x$ be the current
  vertex and $r$ its distance to the root $w$. For every vertex $y$ in
  the subtree of $T$ rooted at $x$, we set $M(y,r):=x$. Assuming that
  the table $M$ and the distance matrix $D:=(d(u,v): u,v\in X)$
  between the vertices of $G$ are available, we can compute $x_y =
  M(x,\lfloor(x|y)_w\rfloor)$, $y_x = M(y,\lfloor(x|y)_w\rfloor)$ and
  $d(x_y,y_x)$ in constant time for each pair of vertices
  $x,y$, and thus $\rho = \max\{d(x_y,y_x) : x,y \in V\}$ can be
  computed in $O(n^2)$ time.
\end{proof}

\medskip
Theorem \ref{main}  provides a new characterization of infinite
hyperbolic graphs.
\begin{corollary} Consider an infinite graph $G$ and an arbitrary BFS-tree $T$
  rooted at a vertex $w$. The graph $G$ is hyperbolic if and only if
  its rooted insize $\rho_{w,T}(G)$ is finite.
\end{corollary}

When the graph $G$ is given by its adjacency list, one can compute its
distance-matrix in $O(\min(mn,n^{2.38}))$ time and then use
the algorithm described in the proof of Theorem \ref{main}.
However, we explain in the next
proposition how to obtain an $8$-approximation of $\delta(G)$ in
$O(mn)$ time using only linear space.

\begin{proposition}\label{prop-algo-sans-M}
  For any graph $G$ with $n$ vertices and $m$ edges that is given by
  its adjacency list, one can compute an $8$-approximation (with an
  additive constant 1) of the hyperbolicity $\delta(G)$ of $G$ in
  $O(mn)$ time and in linear $O(n+m)$ space.
\end{proposition}

\begin{proof}
  Fix a vertex $w$ and compute a BFS-tree $T$ of $G$ rooted at
  $w$. Note that at the same time, we can compute the value $d(w,x)$
  for each $x \in V$.

  For each vertex $x$, consider the map $P_x:\{0,\ldots,d(w,x)\} \to
  V$ such that for each $0 \leq i \leq d(w,x)$, $P_x(i)$ is the unique
  vertex on the path from $w$ to $x$ in $T$ at distance $i$ from $w$.
  For every vertex $x$, consider the map $Q_{x}: V \to \N \cup \{\infty\}$
  such that for each $y \in V$, $Q_x(y) = d(y,P_x(i))$ if $i = d(w,y)
  \leq d(w,x)$ and $Q_x(y) = \infty$ otherwise.

  We perform a depth first traversal of $T$ starting at $w$ and
  consider every vertex $x$ in this order. Initially, $P_x = P_w$ can
  be trivially computed in constant time and $Q_x = Q_w$ can be
  initialized in $O(n)$ time.  During the depth first traversal of
  $T$, each time we go up or down, $P_x$ can be updated in constant
  time.  Assume now that a vertex $x$ is fixed. In $O(n+m)$ time and space, we
  compute $d(x,y)$ for every $y \in V$ by performing a BFS of $G$ from
  $x$.  Moreover, each time we modify $x$, for each $y$, we can update
  $Q_{x}(y)$ in constant time by setting $Q_x(y): = \infty$ if $d(w,y)
  > d(w,x)$, setting $Q_{x}(y): = d(x,y)$ if $d(w,y) = d(w,x)$, and keeping the previous value
  if $d(w,y)<d(w,x)$.

  We perform a depth first traversal of $T$ from $w$ and consider
  every vertex $y$ in this order. As for $P_x$, we can update $P_y$ in
  constant time at each step.  Since $d(w,x), d(w,y)$, and $d(x,y)$
  are available, one can compute $(x|y)_w$ in constant
  time. Therefore, in constant time, we can find $y_x = P_y(\lfloor
  (x|y)_w \rfloor)$ using $P_y$ and compute $d(x_y,y_x) = Q_x(y_x)$
  using $Q_x$.

  Consequently, for each $x$, we compute $\max \{d(x_y,y_x) : y \in
  V\}$ in $O(m)$ time and therefore, we compute $\rho_{w,T}(G)$ in
  $O(mn)$ time. At each step, we only need to store the distances from
  all vertices to $w$ and to the current vertex $x$, as well as
  arrays representing the maps $P_x, Q_x$, and $P_y$. This can be done
  in linear space.
\end{proof}

\begin{remark}
  If we are given the distance-matrix $D$ of $G$, we can use the
  algorithm described in the proof of
  Proposition~\ref{prop-algo-sans-M} to avoid using the $O(n^2)$ space
  occupied by table $M$ in the proof of Theorem \ref{main}. 
  In this
  case, since the distance-matrix $D$ of $G$ is available, we do not
  need to perform a BFS for each vertex $x$ and the algorithm computes
  $\rho_{w,T}(G)$ in $O(n^2)$ time.
\end{remark}

The following result shows that the  bounds  in Theorem \ref{main} are optimal.

\begin{proposition} \label{sharp1} For any positive integer $k$,
  there exists a graph $H_k$, a vertex $w$, and a BFS-tree $T$ rooted
  at $w$ such that $\delta(H_k)=k$ and $\rho_{w,T}(H_k)=4k$.

  For any positive integer $k$, there exists a graph $G_k$, a vertex
  $w$, and a BFS-tree $T$ rooted at $w$ such that $\rho_{w,T}(G_k)\le
  2k$ and $\delta(G_k)=4k$.
\end{proposition}

\begin{proof} The graph $H_k$ is the $2k\times 2k$ square grid from
  which we removed the vertices of the rightmost and downmost
  $(k-1)\times (k-1)$ square (see Fig.~\ref{fig-3}, left).  The graph
  $H_k$ is a median graph and therefore its hyperbolicity is the size
  of a largest isometrically embedded square subgrid
  \cite{Chepoi08,Ha}. The largest square subgrid of $H_k$ has size
  $k$, thus $\delta(H_k)=k$.

  Let $w$ be the leftmost upmost vertex of $H_k$. Let $x$ be the
  downmost rightmost vertex of $H_k$ and $y$ be the rightmost downmost
  vertex of $H_k$. Then $d(x,y)=2k$ and $d(x,w)=d(y,w)=3k$. Let $P'$
  and $P''$ be the shortest paths between $w$ and $x$ and $w$ and $y$,
  respectively, running on the boundary of $H_k$.  Let $T$ be any
  BFS-tree rooted at $w$ and containing the shortest paths $P'$ and
  $P''$.  The vertices $x_y\in P'$ and $y_x\in P''$ are located at
  distance $(x|y)_w = \frac{1}{2}(d(w,x)+d(w,y)-d(x,y)) = 2k$ from
  $w$. Thus $x_y$ is the leftmost downmost vertex and $y_x$ is the
  rightmost upmost vertex. Hence $\rho_{w,T}(H_k) \geq
  d(x_y,y_x)=4k$. Since the diameter of $H_k$ is $4k$, we conclude
  that $\rho_{w,T}(H_k) =4k = 4\delta(H_k)$.

\begin{figure}
\begin{center}
\begin{tabular}{ccc}
\scalebox{0.31}{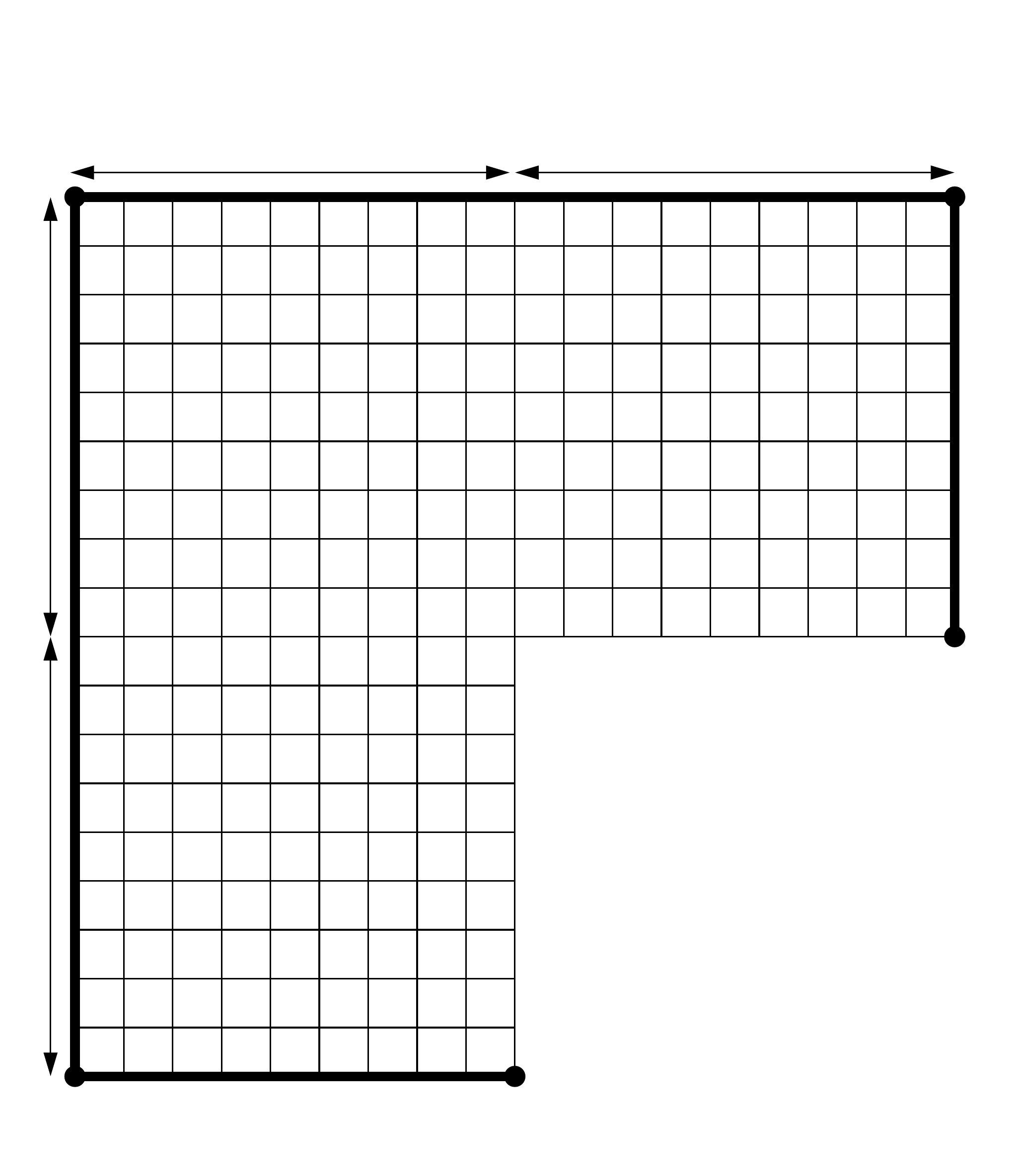} & \hfill & \scalebox{0.23}{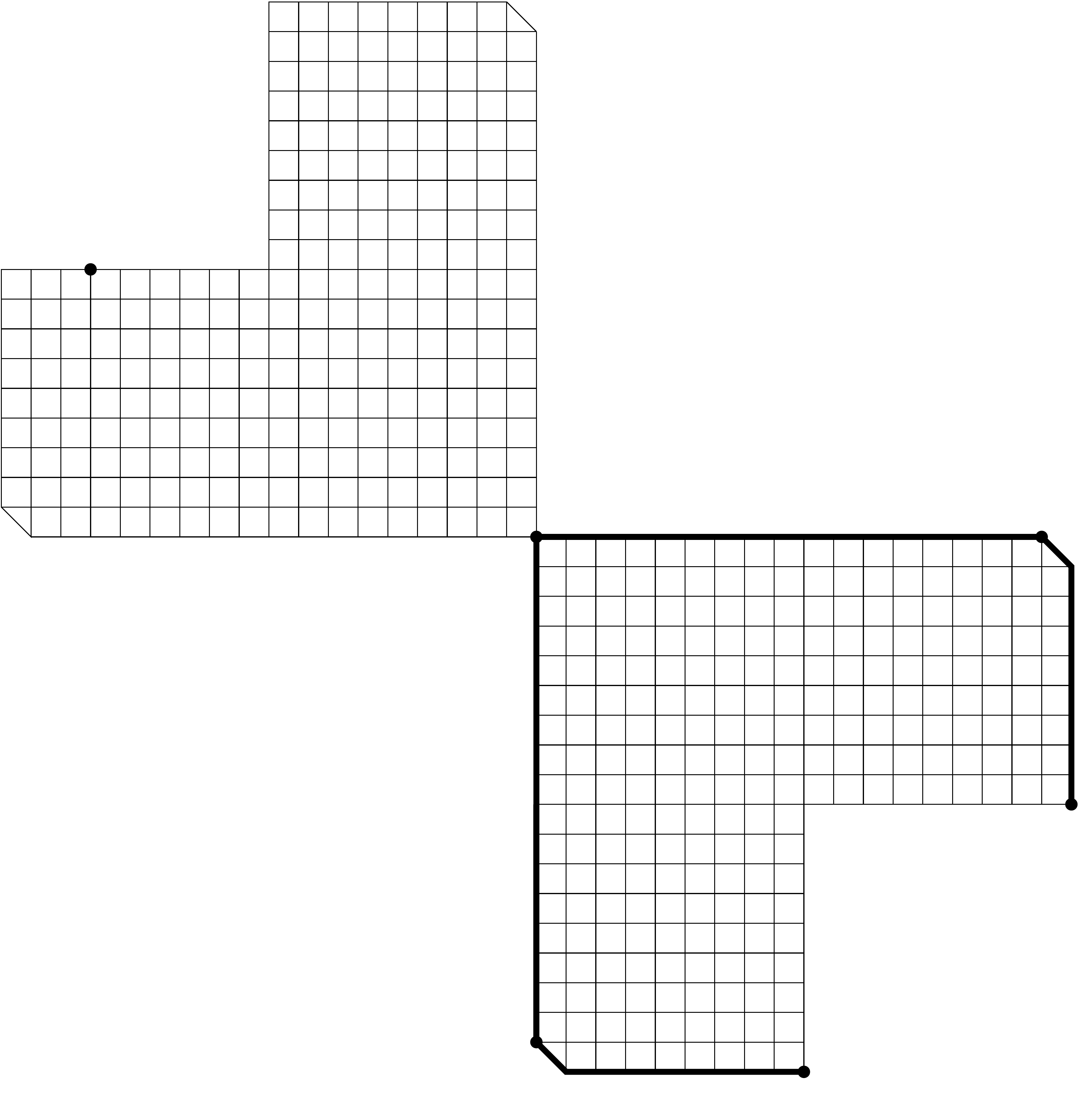}\\
\\
$H_k$ & ~~~~ & $H^*_k$
\end{tabular}
\end{center}
\caption{Since $\rho_{w,T}(H_k) = d(x_y,y_x) = 4k = 4 \delta(H_k)$,
  the inequality $\rho_{w,T}(G) \leq 4\delta$ is tight in the proof of
  Theorem~\ref{main}. Since
  $\rho_{w^*,T}(H^*_k) \geq 4k-2 = 4\delta(H_k^*) - O(1)$ for any
  $w^*,T$, we have $\rho_{-}(H_k^*)\geq 4\delta(H_k^*) -
  O(1)$. } \label{fig-3}
\end{figure}

Let $G_k$ be the $4k\times 4k$ square grid and note that
$\delta(G_k)=4k$. Let $w$ be the center of $G_k$. 
We suppose that $G_k$ is isometrically embedded in
the $\ell_1$-plane in such a way that $w$ is mapped to the origin of
coordinates $(0,0)$ and the four corners of $G_k$ are mapped to the
points with coordinates $(2k,2k),(-2k,2k),(-2k,-2k),(2k,-2k)$, We
build the BFS-tree $T$ of $G_k$ as follows.  First we connect $w$ to
each of the corners of $G_k$ by a shortest zigzagging path (see
Fig.~\ref{fig-4}). For each $i \leq k$, we add a
vertical path from $(i,i)$ to $(i,2k)$, from $(i,-i)$ to $(i,-2k)$,
from $(-i,i)$ to $(-i,2k)$, and from $(-i,-i)$ to
$(-i,-2k)$. Similarly, for each $i \leq k$, we add a horizontal
path from $(i,i)$ to $(2k,i)$, from $(i,-i)$ to $(2k,-i)$, from
$(-i,i)$ to $(-2k,i)$, and from $(-i,-i)$ to $(-2k,-i)$.  For any
vertex $v=(i,j)$, the shortest path of $G_k$ connecting $w$ to $v$ in
$T$ has the following structure: it consists of a subpath of one of
the zigzagging paths until this path arrives to the vertical or
horizontal line containing $v$ and then it continues along this line
until $v$.

We divide the grid in four \emph{quadrants} $Q_1 = \{(i,j) : 0 \leq i,j \leq 2k\}$, $Q_2  = \{(i,j) : -2k \leq i \leq 0,0 \leq j \leq 2k\}$, $Q_3 = \{(i,j) : -2k \leq i,j \leq 0\}$ and $Q_4 = \{(i,j) : 0 \leq i \leq 2k, -2k \leq j \leq 0\}$.
Pick any two vertices $x=(i,j)$ and $y=(i',j')$. If $x$ and $y$ belong to opposite quadrants of $G_k$, then $w\in I(x,y)$ and $x_y=y_x=w$. So, we can suppose that either $x$ and $y$ belong to the same quadrant or to two incident quadrants of $G_k$. Denote by $m=m(x,y,w)$ the \emph{median} of the triplet $x,y,w$, {i.e.}, the unique vertex in the intersection
$I(x,y)\cap I(x,w)\cap I(y,w)$ ($m$ is the vertex having the median element of the list $\{ i,0,i'\}$ as the first coordinate and the median element of the list $\{ j,0,j'\}$ as the second coordinate). Notice that $m$ has the same distance $r:=(x|y)_w$ to $w$ as $x_y$ and $y_x$ ($(x|y)_w$ is integer because $G_k$ is bipartite).

\begin{figure}
\begin{center}
\scalebox{0.35}{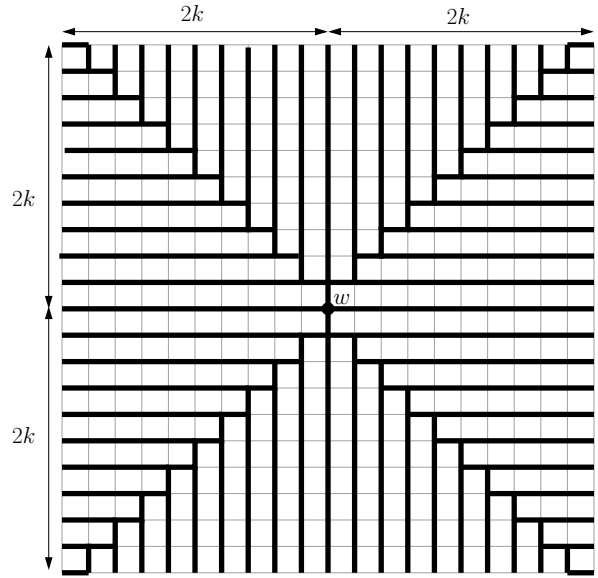}
\end{center}
\caption{Since $\rho_{w,T}(G_k) \leq 2k = \frac{1}{2}\delta(G_k)$, the
  inequality $\delta \leq 2\rho_{w,T}(G) +1$ is tight (up to an additive factor
  of $1$) in the proof of Theorem~\ref{main}.} \label{fig-4}
\end{figure}

\begin{case}
$x=(i,j)$ and $y=(i',j')$ belong to the same quadrant of $G_k$.
\end{case}

Suppose that $x$ and $y$ belong to the first quadrant (alias $2k\times
2k$ square) $Q_1$ of $G_k$, {i.e.}, $i,j,i',j'\ge 0$.  We divide $Q_1$
into four $k\times k$ squares $Q_{11} =\{ (i,j): k \leq i,j \leq
2k\}$, $Q_{12} =\{ (i,j): 0 \leq i \leq k, k \leq j \leq 2k\}$,
$Q_{13} =\{ (i,j): 0 \leq i,j \leq k\}$ and $Q_{14} =\{ (i,j): k
\leq i \leq 2k, 0 \leq j \leq k\}$.

Since the vertices $x_y,y_x$, and $m$ have the same distance $r$ to
$w$ and belong to $Q_1$, they all belong to the same side $L$ of the
sphere $S_{r}(w)$ of the $\ell_1$-plane of radius $r$ and centered at
$w$. Let $L_0:=[x_y,y_x]$ be the subsegment of $L$ between $x_y$ and
$y_x$.
Notice first that if $L_0$ is completely contained in one or two
incident $k\times k$ squares (say in $Q_{12}$ and $Q_{13}$), then
$d(x_y,y_x)\le 2k$. Indeed, in this case $L_0$ can be extended to a
segment $L'_0$ having its ends on two vertical sides of the rectangle
$Q_{12}\cup Q_{13}$. Therefore, $L'_0$ is the diagonal of a $k\times
k$ square included in $Q_{12}\cup Q_{13}$, thus the $\ell_1$-length of
$L'_0$ (and thus of $L_0$) is at most $2k$. Thus we can suppose that
the vertices $x_y$ and $y_x$ are located in two non incident $k\times
k$ squares. This is possible only if one of these vertices belongs to
$Q_{12}$ and another belongs to $Q_{14}$, say $x_y\in Q_{12}$ and
$y_x\in Q_{14}$. This implies that $x\in Q_{11}\cup Q_{12}$ and $y\in
Q_{11}\cup Q_{14}$. Notice that neither $x$ nor $y$ may belong to
$Q_{11}$. Indeed, if $x\in Q_{11}$, then the center $(k,k)$ of $Q_1$
belongs to the path of $T$ from $w$ to $x$. Consequently, this path is
completely contained in $Q_{11}\cup Q_{13}$, contrary to the
assumption that $x_y\in Q_{12}$. Thus $x\in Q_{12}$ and $y\in Q_{14}$,
{i.e.}, $0 \le i \le k, k \le j \le 2k, k\le i'\le 2k,$ and $0 \le j'\le
k$.  This means that the median $m$ of the triplet $x,y,w$ has
coordinates $(i,j')$ and belongs to $Q_{13}$. 
The path of $T$ from $w$
to $x=(i,j)$ is zigzagging until $(i,i)$ and then is going
vertically. Analogously, the path of $T$ from $w$ to $y=(i',j')$ is
zigzagging until $(j',j')$ and then is going horizontally. If we
suppose, without loss of generality, that $i\le j'$, then $m=(i,j')$
belongs to the $(w,x)$-path of $T$ and therefore $x_y=m$. This
contradicts our assumption that $x_y$ and $y_x$ do not belong to a
common or incident $k\times k$ squares. This concludes the proof of
Case 1.

\begin{figure}
\begin{center}
\begin{tabular}{ccc}
\scalebox{0.30}{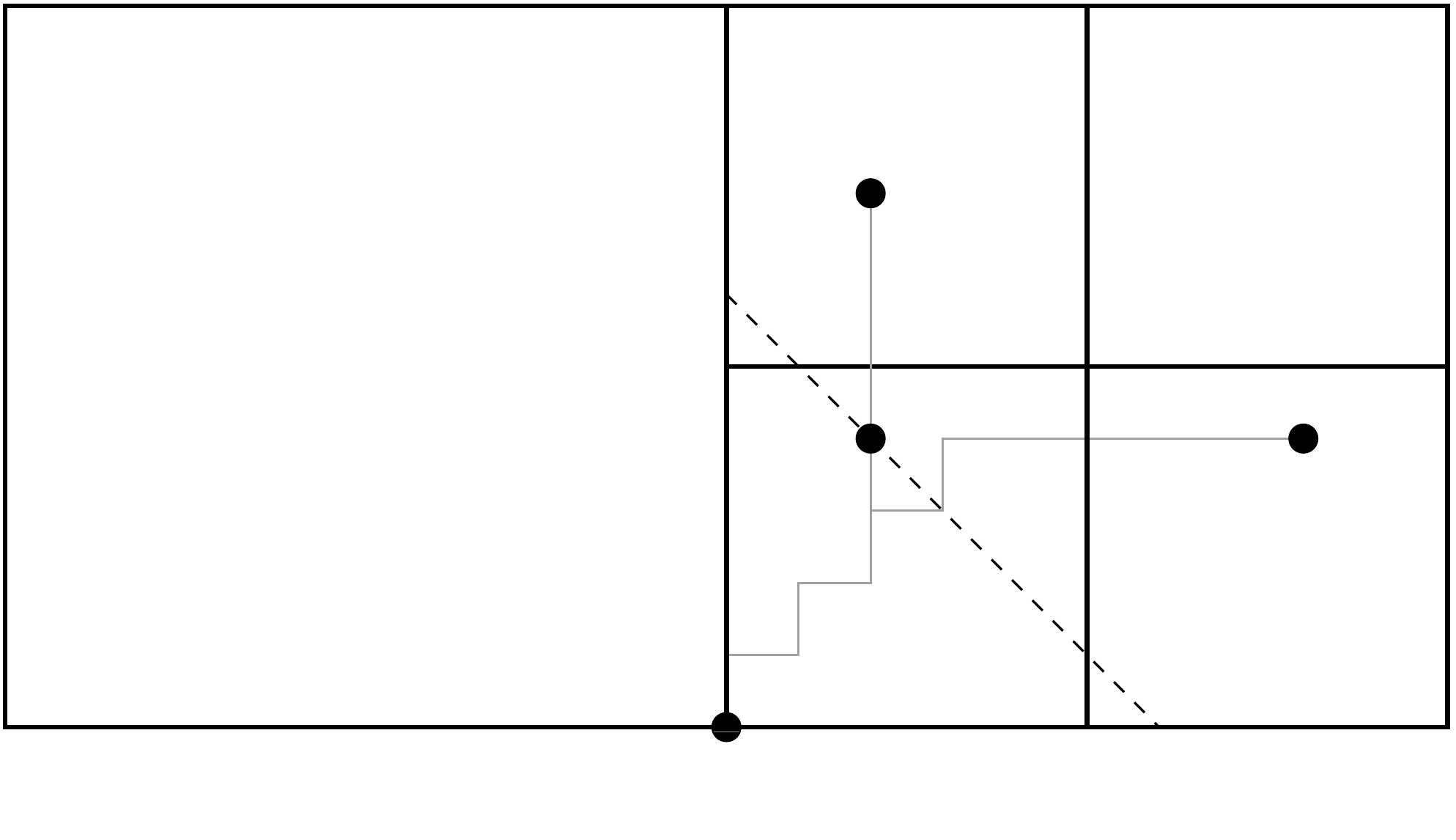} & \hspace{0.7cm} & \scalebox{0.30}{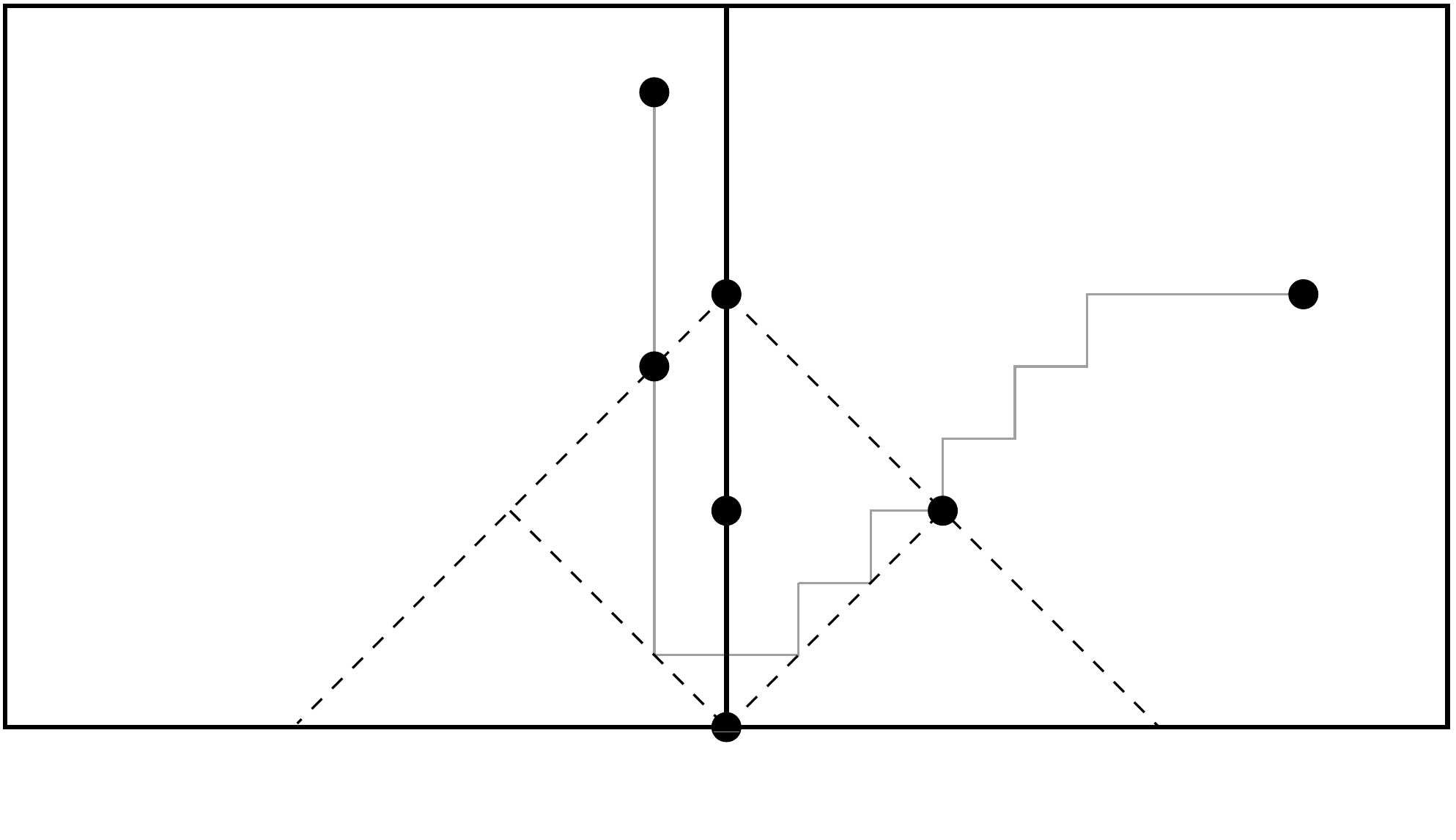}\\
\\
Case 1 & ~~~~ & Case 2
\end{tabular}
\end{center}
\caption{To the proof of the second statement of Proposition~\ref{sharp1}.} \label{fig-5}
\end{figure}

\begin{case}
$x=(i,j)$ and $y=(i',j')$ belong to incident quadrants of $G_k$.
\end{case}

Suppose that $x\in Q_1$ and $y\in Q_2$, {i.e.},
$i,j,j'\ge 0$ and $i'\le 0$. The points $x_y$ and $y_x$ belong to
different but incident sides $L,L'$ of the sphere $S_{r}(w)$ of the
$\ell_1$-plane, $x_y\in L$ and $y_x\in L'$. The median point $m$ also
belongs to these sides. Since $i'\le 0\le i$, we conclude that $m$ has
$0$ as the first coordinate. Thus $m$ belongs to both segments $L$ and
$L'$. Suppose without loss of generality that $j\le j'$, {i.e.}, the
second coordinate of $m$ is $j$. Consequently, $r=j$. If $i\ge
\lfloor\frac{j}{2}\rfloor$, then the vertex $(\lfloor \frac{j}{2}
\rfloor ,\lfloor \frac{j}{2} \rfloor)$ belongs simultaneously to $L$
and to the path of $T$ connecting $w$ and $x$; thus in this case $x_y$
is either $(\lfloor \frac{j}{2} \rfloor, \lceil \frac{j}{2} \rceil)$
or $(\lceil \frac{j}{2} \rceil, \lfloor \frac{j}{2} \rfloor)$.  If
$i<\lfloor \frac{j}{2} \rfloor$, then one can easily see that the
intersection of $L$ with the path of $T$ from $w$ to $x$ is the vertex
$x_y=(i,j-i)$. In both cases, $d(x_y,c) \leq \lceil \frac{j}{2}
\rceil$ where $c = (0,\lfloor \frac{j}{2} \rfloor)$. Analogously, we
can show that $d(y_x,c) \leq \lceil \frac{j}{2} \rceil$.
Consequently, $d(x_y,y_x) \leq d(x_y,c)+d(c,y_x) = 2 \lceil
\frac{j}{2} \rceil \leq 2k$ as $j \leq 2k$.  This finishes the
analysis of Case 2. Consequently, $\rho=\rho_{w,T}(G_k)\le 2k$,
concluding the proof  of the proposition.
\end{proof}

The definition of $\rho_{w,T}(G)$ depends on the choice of the basepoint
$w$ and of the BFS-tree $T$ rooted at $w$.  We show below that the
best choices of $w$ and $T$ do not improve the bounds in Theorem
\ref{main}.  For a graph $G$, let ${\rho}_{-}(G)=\min\{ \rho_{w,T}(G):
w\in V \text{ and } T \text{ is a BFS-tree rooted at } w\}$ and call
${\rho}_{-}(G)$ the \emph{minsize} of $G$.  On the other hand, the {\it
  maxsize} $\rho_{+}(G)=\max\{ \rho_{w,T}(G): w\in V \text{ and } T
\text{ is a BFS-tree rooted at } w\}$ of $G$ coincides with its insize
$\iota(G)$. Indeed, from the definition, $\rho_{+}(G)\le
\iota(G)$. Conversely, consider a geodesic triangle $\Delta(x,y,w)$
maximizing the insize and suppose, without loss of generality, that
$d(x_y,y_x)=\iota(G)$, where $x_y$ and $y_x$ are chosen on the sides
of $\Delta(x,y,w)$. Then, if we choose a BFS-tree rooted at $w$, and
such that $x_y$ is an ancestor of $x$ and $y_x$ is an ancestor of $y$,
then one obtains that $\rho_{+}(G)\ge \iota(G)$. We show in
Section~\ref{exact} that $\rho_+(G)$ ($ = \iota(G) = \tau(G)$) can be computed in
polynomial time, and by Proposition~\ref{prop:sl-vs-th}, it gives a
$4$-approximation of $\delta(G)$.

On the other hand, the next proposition shows that one cannot get
better than a factor 8 approximation of hyperbolicity if instead of
computing $\rho_{w,T}(G)$ for an arbitrary BFS-tree $T$ rooted at some
arbitrary vertex $w$, we compute the minsize $\rho_{-}(G)$.
Furthermore, we show in Section~\ref{exact} that we cannot approximate
$\rho_{-}(G)$ with a factor strictly better than 2 unless P = NP.

\begin{proposition} \label{sharp2} For any positive integer $k$, there
  exists a graph $H^*_k$ with $\delta(H^*_k)=k+O(1)$ and
  $\rho_{+}(H^*_k)\ge\rho_{-}(H^*_k)\ge 4k-2$ and a graph $G^*_k$ with
  $\delta(G^*_k)=4k$ and $\rho_{-}(G^*_k)\le 2k$.
\end{proposition}

\begin{proof} The graph $G^*_k$ is just the graph $G_k$ from
  Proposition~\ref{sharp1}.  By this proposition and the definition of
  $\rho_{-}(G^*_k)$, we have $\delta(G^*_k)=4k$ and
  $\rho_{-}(G^*_k)\le \rho_{w,T}(G^*_k)\le 2k$.  Let $H'_k$ be the graph
  $H_k$ from Proposition~\ref{sharp1} in which we cut-off the vertices
  $x_y$ and $y_x$: namely, we removed these two vertices and made
  adjacent their neighbors in $H_k$. This way, in $H'_k$ the vertices
  $x,y$ are pairwise connected to $w$ by unique shortest paths, that
  are the boundary paths $P'$ and $P''$ of $H_k$ shortcut by removing
  $x_y$ and $y_x$ and making their neighbors adjacent. Since
  $\delta(H_k)=k$, from the definition of $H'_k$ it follows that
  $\delta(H'_k)$ may differ from $k$ by a small constant. Let $H^*_k$
  be the graph obtained by gluing two copies of $H'_k$ along the
  leftmost upmost vertex $w$ (see Fig.~\ref{fig-3}, right). Consequently,
  the vertex $w$ becomes the unique
  articulation point of $H^*_k$ which has two blocks, each of them
  isomorphic to $H'_k$. Pick any
  basepoint $w^*$ and any BFS-tree
  $T^*$ of $H^*_k$ rooted at $w^*$.  We assert that $\rho_{w^*,T^*}(H^*_k)\ge
  4k-2$. Indeed, pick the vertices $x$ and $y$ in the same copy of
  $H'_k$ that do not contain $w^*$ (if $w^* \neq w$). Then both paths
  of $T^*$ connecting $w^*$ to $x$ and $y$ pass through  the vertex
  $w$. Since $w$ is connected to $x$ and $y$ by unique shortest paths
  $P'$ and $P''$, the paths $P'$ and $P''$ belong to $T^*$. The
  vertices $x^*_y$ and $y^*_x$ in the tree $T^*$ are the vertices of
  $P'$ and $P''$, respectively, which are the neighbors of $x_y$ and
  $y_x$ located at distance $\lfloor(x|y)_w\rfloor$ from $w$. One can
  easily see that $d(x^*_y,y^*_x)=4k-2$, {i.e.}, $\rho_{w^*,T^*}(H^*_k)\ge
  4k-2$.
\end{proof}

If instead of knowing the distance-matrix $D$, we only know the
distances between the vertices of $G$ up to an additive error $k$,
then we can define a parameter $\widehat{\rho}_{w,T}(G)$ in a similar
way as the rooted insize $\rho_{w,T}(G)$ is defined and show that
$2\widehat{\rho}_{w,T}(G) + k + 1$ is an 8-approximation of
$\delta(G)$ with an additive error of $3k +1$.

\begin{proposition}\label{prop-dist-approx-hyp-approx}
  Given a graph $G$ with $n$ vertices, a BFS-tree $T$ rooted at a
  vertex $w$, and a matrix $\widehat{D}$ such that
  $d(x,y) \leq \widehat{D}(x,y) \leq d(x,y)+k$, we can compute in time
  $O(n^2)$ a value $\widehat{\rho}_{w,T}(G)$ such that
  $\delta(G) \leq 2\widehat{\rho}_{w,T}(G) + k + 1 \leq 8 \delta(G) +
  3k +1 $.
\end{proposition}

\begin{proof}
  Consider a graph $G=(V,E)$ with $n$ vertices, a vertex $w\in V$, and
  a BFS-tree of $G$ rooted at $w$.  We can assume that the exact
  distance $d(x,w)$ in $G$ from $w$ to every vertex $x\in V$ is
  known. For any vertex $x\in V$, let $[w,x]_T$ be the path connecting
  $w$ to $x$ in $T$.  Denote by $x_y$ the point of $[w,x]_T$ at
  distance $\lfloor(\widehat{x|y})_w\rfloor$ from $w$ and by $y_x$ the
  point of $[w,y]_T$ at distance $\lfloor(\widehat{x|y})_w\rfloor$
  from $w$, where
  $(\widehat{x|y})_w:=\frac{1}{2}(d(x,w)+d(y,w)-\widehat{D}(x,y))$.
  Let
  $\wrho :=\wrho_{w,T}(G):= \max\{\widehat{D}(x_y,y_x): x,y\in V\}$.
  Using the same arguments as in the proof of Theorem~\ref{main}, if
  $\widehat{D}(x,y)$ is known for each $x,y\in V$, the value of
  $\wrho$ can be computed in $O(n^2)$ time.  In what follows, we show
  that $\delta(G)\le 2\wrho + k + 1 \le 8 \delta(G) + 3k + 1$.

  Let $\delta:=\delta(G)$, $\delta_w:=\delta_w(G)$, and $\tau :=
  \tau(G)$. By Proposition \ref{hyp_basepoint}, $\delta\le 2\delta_w$,
  and by Proposition \ref{prop:sl-vs-th}, $\tau \leq 4\delta$.
  We proceed in two steps: in the first step, we show that $\wrho\le
  \tau+k \le 4\delta+k$, in the second step, we prove that
  $\delta_w\le \wrho + \frac{k+1}{2}$. Hence, combining both steps we
  obtain $\delta\le 2\delta_w \le 2 \wrho+ k+1 \le 8\delta+3k+ 1$.

  The first assertion follows from the fact that for any $x,y \in V$,
  $\lfloor(\widehat{x|y})_w\rfloor\le (\widehat{x|y})_w\le (x|y)_w$
  (as $d(x,y)\leq \widehat{D}(x,y)$). Consequently, we have
  ${d}(x_y,y_x)\le \tau$ and therefore $\widehat{D}(x_y,y_x)\le
  {d}(x_y,y_x)+k\le \tau+k \le 4\delta+k$.

  To prove that $\delta_w \le \wrho+ \frac{k+1}{2}$, for any
  quadruplet $x,y,z,w$ containing $w$, we show the four-point
  condition $d(x,z)+d(y,w)\le \max\{
  d(x,y)+d(z,w),d(y,z)+d(x,w)\}+(2\wrho+k+1)$. Assume without loss of
  generality that $d(x,z)+d(y,w) \ge
  \max\{d(x,y)+d(z,w),d(y,z)+d(x,w)\}$ and that $d(w,x_y) = d(w,y_x)
  \le d(w,y_z) = d(w,z_y)$. The remaining part of the proof closely
  follows the proof of Theorem~\ref{main}.

  From the definition of $\wrho$, $d(x_y,y_x)\le \wrho$ and
  $d(y_z,z_y)\le \wrho$.  Consequently, by the definition of $x_y, y_x,
  y_z, z_y$ and by the triangle inequality, we get
  \begin{align*}
    d(y,w)+d(x,z)& \le
    d(y,w)+d(x,x_y)+d(x_y,y_x)+d(y_x,y_z)+d(y_z,z_y)+d(z_y,z)\\
    &\le (d(y,y_z)+d(y_z,w))+d(x,x_y)+\wrho+
    d(y_x,y_z)+\wrho+d(z_y,z)\\
    & = d(y,y_z)+d(w, z_y) + d(x,x_y)+ d(y_x,y_z)+ d(z_y,z)+ 2
    \wrho\\
    &= d(y,y_z)+d(x,x_y)+ (d(y,y_x)-d(y,y_z))+d(w,z)+2 \wrho \\
    &= d(x,x_y)+d(y,y_x) + d(z,w) + 2 \wrho \\
    & \leq d(x,y) + d(z,w)+ k+1 + 2\wrho.
  \end{align*}
The last line inequality follows from
\begin{align*}
  d(x,x_y)+d(y,y_x) &= d(x,w)-\lfloor(\widehat{x|y})_w\rfloor +
  d(y,w) - \lfloor(\widehat{x|y})_w\rfloor\\
      &\le d(x,w)+ d(y,w)-2(\widehat{x|y})_w+1\\
      &= d(x,w) +d(y,w) - (d(x,w) + d(y,w) - \widehat{D}(x,y)) +1\\
      &\le d(x,y)+k+1.
\end{align*}
This establishes the four point condition for $w,x,y,z$ and proves
that $\delta_w \leq \wrho + \frac{k +1}{2}$.
\end{proof}

\begin{remark}
  A consequence of Proposition~\ref{prop-dist-approx-hyp-approx}
  (suggested by one of the referees) is that for any two graphs $G, H$ on
  the same set of vertices $V$, if 
  $\max \{|d_G(x,y) - d_H(x,y)| : x,y \in V\} \leq k$, then 
  $\delta(G)$ can be bounded linearly by a function of $\delta(H)$ and
  $k$. This property can be viewed as a specific instance of the fact
  that hyperbolicity is a quasi-isometry invariant~\cite{BrHa}.
\end{remark}

Interestingly, the rooted insize $\rho_{w,T}(X)$ can also be defined
in terms of a distance approximation parameter. Consider a geodesic
space $X$ and a GS-tree $T$ rooted at some point $w$, and let
$\rho := \rho_{w,T}(X)$.
For a point $x\in X$ and $r\in {\mathbb R}^+$, denote by $x(r)$ the
unique point of $[w,x]_{T}$ at distance $r$ from $w$ if $r<d(w,x)$ and
the point $x$ if $r\ge d(w,x)$. 
For any $x, y$ and $\epsilon \in {\mathbb R}^+$, let
$r_{xy}(\epsilon):=\sup\{ r: d(x(r'),y(r'))\le \epsilon \text{ for any
} 0\le r'\le r \}$. This supremum is a maximum because the function
$r' \mapsto d(x(r'),y(r'))$ is continuous.  Observe that by
Proposition~\ref{prop:rhomu}, $\rho = \inf\{\epsilon :
r_{xy}(\epsilon) \geq (x|y)_w \text{ for all } x,y \}$.

Denote by $x_y(\epsilon)$ (respectively, $y_x(\epsilon)$) the point of
$[x,w]_{T}$ (respectively, of $[w,y]_{T}$) at distance
$r_{xy}(\epsilon)$ from $w$.  Let
$\widehat{d}_{\epsilon}(x,y):=d(x,x_y(\epsilon))+\epsilon+d(y_x({\epsilon}),y)$.
By the triangle inequality, $d(x,y)\le
d(x,x_y(\epsilon))+d(x_y(\epsilon),y_x(\epsilon))+d(y_x({\epsilon}),y)\le
\widehat{d}_{\epsilon}(x,y)$.  Observe that for any $\epsilon$ and for
any $x,y$, we have $r_{xy}(\epsilon) \geq (x|y)_w$ if and only if
$d(x,x_y(\epsilon))+d(y_x({\epsilon}),y)\le d(x,y)$, {i.e.}, if and only
if $d(x,y) \leq \widehat{d}_{\epsilon}(x,y) \leq d(x,y)
+\epsilon$. Consequently, $\rho = \inf\{\epsilon: d(x,y) \leq
\widehat{d}_{\epsilon}(x,y) \leq d(x,y) +\epsilon \text{ for all } x,y
\}$.

When we consider a graph $G$ with a BFS-tree $T$ rooted at some vertex
$w$, we have similar results for $\rho:= \rho_{w,T}(G)$.  For a vertex
$x$, we define $x(r)$ as before when $r$ is an integer and for
vertices $x,y$, we define
$r_{xy}(\epsilon):=\max\{ r \in \mathbb{N}: d(x(r'),y(r'))\le \epsilon
\text{ for any } 0\le r'\le r \}$. Since
$\rho = \inf\{\epsilon: r_{xy}(\epsilon) \geq \lfloor (x|y)_w \rfloor
\text{ for all } x,y\}$, we get that
$ d(x,y) \leq \widehat{d}_{\rho}(x,y)+1 \leq d(x,y) +\rho +1$.

The $k$th power $G^k$ of a graph $G$ has the same vertex set as $G$
and two vertices $u,v$ are adjacent in $G^k$ if $d(u,v) \leq k$.  With
$G^k$ at hand, for a fixed vertex $x\in V$ the values of $r_{xy}(k)$
and $\widehat{d}_{k}(x,y)$, for every $y\in V$, can be computed in
linear time using a simple traversal of the BFS-tree $T$. Consequently, we obtain the following result.

\begin{proposition}
  If the distance matrix $D$ of a graph $G$ is unknown but the $k$th
  power graph $G^k$ of $G$ is given for $k \geq \rho_{w,T}(G)$, then one
  can approximate the distance matrix $D$ of $G$ in optimal $O(n^2)$
  time with an additive term depending only on $k$.
\end{proposition}

\subsection{Fast approximation of thinness, slimness, and insize}
Using Proposition~\ref{prop:sl-vs-th}, Theorem~\ref{main}, and
Proposition~\ref{prop-algo-sans-M}, we get the following corollary.

\begin{corollary} \label{cr:appr-th} For a graph $G$ and a BFS-tree
  $T$ rooted at a vertex $w$, $\tau(G)\le 8\rho_{w,T}(G)+4\le 8
  \tau(G)+4$ and $\varsigma(G)\le 6\rho_{w,T}(G)+3\le 24\varsigma(G)+3$.
  Consequently, an 8-approximation (with additive surplus 4) of the
  thinness $\tau(G)$ and a 24-approximation (with additive surplus 3)
  of the slimness $\varsigma(G)$ can be found in $O(n^2)$ time
  (respectively, in $O(mn)$ time) for any graph $G$ given by its
  distance matrix (respectively, its adjacency list).
\end{corollary}

\begin{proof} Indeed, $\tau(G)=\iota(G)\leq 4 \delta(G)\le 8\rho_{w,T}(G)+4\le 8 \iota(G)+4=8\tau(G)+4$.
Since $\varsigma(G)\le 3\delta(G)+1/2$, $\delta(G)\le 2\rho_{w,T}(G)+1$ and
$\varsigma(G)$ is an integer, we get $\varsigma(G)\le 6\rho_{w,T}(G)+3$. Hence,
$\varsigma(G)\le 6\rho_{w,T}(G)+3\le 6 \iota(G)+3 \le 24\varsigma(G)+3$.
\end{proof}

In fact, with $\rho_{w,T}(G)$ at hand we can compute a 7-approximation of the
thinness $\tau(G)$ of $G$.

\begin{theorem}\label{th-betterapproxtau}
  Given a graph $G$ (respectively a geodesic metric space $X$) and a BFS-tree
  $T$ (respectively, a GS-tree $T$) rooted at $w$, $\tau(G) \leq
  7\rho_{w,T}(G) + 4 \leq 7 \tau(G) +4$ (respectively, $\tau(X) \leq
  7\rho_{w,T}(X) \leq 7 \tau(X)$). Consequently, a $7$-approximation (with an additive constant 4) of
  the thinness $\tau(G)$ of $G$ can be computed in $O(n^2)$ time
  (respectively, in $O(mn)$ time) for any graph $G$ given by its
  distance matrix (respectively, by its adjacency list).
\end{theorem}

The second statement of the theorem is a corollary of the
first statement, of Theorem~\ref{main}, and of
Proposition~\ref{prop-algo-sans-M}.  To prove the first statement,
we first need the following simple lemma.

\begin{lemma}\label{lem-rhotriangle}
  Given a graph $G$ (respectively, a geodesic metric space $X$) and a BFS-tree
  $T$ (respectively, a GS-tree $T$) rooted at $w$, for any three
  vertices $x,y,z$ such that $z \in I(x,y)$, if $d(y,z) \leq (w|x)_y$
  then $(w|y)_z = (w|x)_z-(w|x)_y+d(y,z) \leq \rho_{w,T}(G) +\frac{1}{2}$
  (respectively, $(w|y)_z = (w|x)_z-(w|x)_y + d(y,z) \leq
  \rho_{w,T}(X)$).
\end{lemma}

\begin{proof}
  Let $\rho := \rho_{w,T}(G)$ and let $[w,x]$, $[w,y]$, $[w,z]$ be the
  three shortest paths from $w$ to respectively $x$, $y$, and $z$ in
  $T$. Let $[x,y]$ be any geodesic going through $z$, and let $[z,x]$
  and $[z,y]$ be the geodesics from $z$ to respectively $x$ and $y$
  that are contained in $[x,y]$. Since $d(x,y) = d(x,z)+d(z,y)$, we
  have $(w|x)_z - (w|x)_y +d(y,z) = (w|y)_z$.

  Pick the vertices $x_z \in [w,x]$, $z_x \in [w,z]$ at distance
  $\lfloor(x|z)_w\rfloor$ from $w$ and $y_z \in [w,y]$, $z_y \in
  [w,z]$ at distance $\lfloor(y|z)_w\rfloor$ from $w$. Notice that
  $(w|x)_z-(w|y)_z = (w|x)_y - d(y,z) \geq 0$.
  Consequently, $z_y \in I(z,z_x)$.  Since $d(x,y) =
  (w|z)_x+(w|x)_z+(w|y)_z+(w|z)_y$ and $\lceil(w|z)_y\rceil -
  \lceil(w|y)_z\rceil = (w|z)_y - (w|y)_z$, we have
  \begin{align*}
     (w|z)_x +(w|x)_z+(w|y)_z+(w|z)_y
    &= d(x,y) \\
    &\leq d(x,x_z)+d(x_z,z_x) +d(z_x,z_y) + d(z_y,y_z)+
    d(y_z,y) \\
    & \leq \lceil (w|z)_x \rceil + \rho + \lceil
    (w|x)_z\rceil -\lceil (w|y)_z \rceil + \rho + \lceil (w|z)_y
    \rceil \\
    & \leq (w|z)_x +\frac{1}{2} + \rho + (w|x)_z
    +\frac{1}{2} - (w|y)_z + \rho + (w|z)_y.
  \end{align*}
  Consequently, $(w|y)_z \leq \rho +\frac{1}{2}$.

  In a geodesic metric space, since $d(x,x_z) = (w|z)_x$, we obtain a similar
  result without the additive constant.
\end{proof}

By definition, $\rho_{w,T}(G) \leq \tau(G)$, thus the first statement of
Theorem~\ref{th-betterapproxtau} follows from the fact that $\iota(G)
= \tau(G)$ and the following proposition.

\begin{figure}
\begin{center}
\includegraphics[scale=0.7]{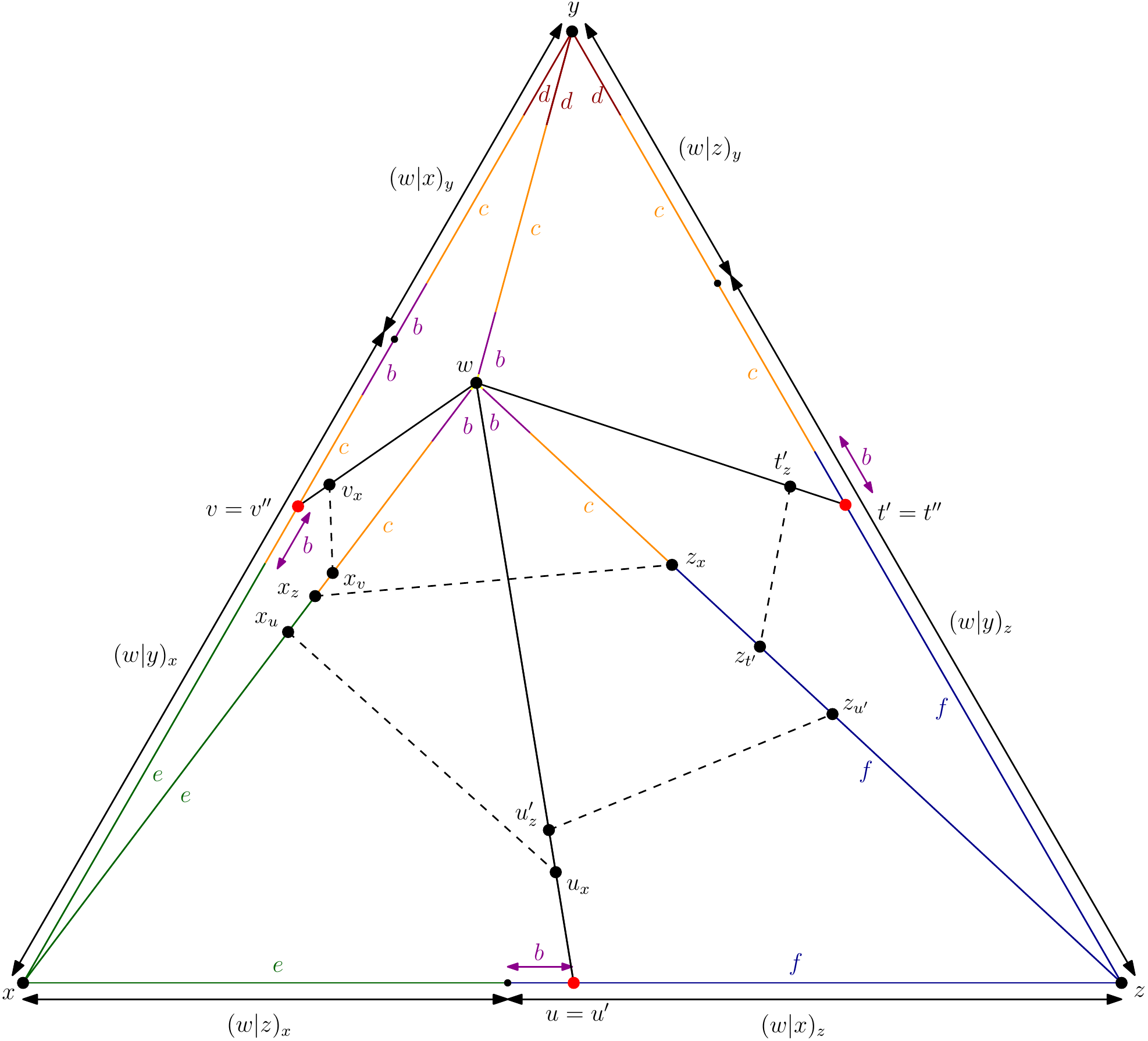}
\end{center}
\caption{To the proof of Proposition~\ref{prop:better-approx-tau}.} \label{fig-approx-tau}
\end{figure}

\begin{proposition} \label{prop:better-approx-tau}
  In a graph $G$ (respectively, a geodesic metric space $X$), for any BFS-tree
  $T$ (respectively, any GS-tree $T$) rooted at some vertex $w$,
  $\iota(G) = \tau(G) \leq 7\rho_{w,T}(G) + 4$ (respectively, $\iota(X) =
  \tau(X) \leq 7\rho_{w,T}(X)$).
\end{proposition}

\begin{proof}
  We prove the proposition for graphs (for geodesic spaces, the proof
  is similar but simpler).  Let $\rho := \rho_{w,T}(G)$.  Consider a
  geodesic triangle $\Delta(x,y,z) = [x,y]\cup[y,z]\cup[z,x]$ and
  assume without loss of generality that $(x|y)_w \leq (y|z)_w \leq
  (x|z)_w$. Let $a := (x|y)_w$, $b:= (y|z)_w - a$, and $c := (x|z)_w -a -
  b$. Let $e := (w|z)_x$, $f := (w|x)_z$ and $d := d(y,w) - (x|z)_w =
  d(y,w) - a - b - c$. See Fig.~\ref{fig-approx-tau} for an illustration
  (in this example, $a$ is very small and not represented in the figure).
  Observe that $d$ may be negative but that
  $a,b,c,e,f \geq 0$. Note that $b \leq \delta_w \leq \rho +
  \frac{1}{2}$ as explained in the proof of Theorem~\ref{main}.
  Observe that $d(w,x) = a+b+c+e$, $d(w,y) = a+b+c+d$, $d(w,z) =
  a+b+c+f$, $d(x,z) = e+f$, $d(y,z) = d+f+2c$, and $d(x,y) =
  e+d+2b+2c.$

  Let $v$ and $u$ be the vertices of $[x,y]$ and $[x,z]$ at distance
  $\lfloor (y|z)_x \rfloor$ from $x$, let $u'$ and $t'$ be the vertices
  of $[z,x]$ and $[z,y]$ at distance $\lfloor (x|y)_z \rfloor$ from $z$,
  and let $v''$ and $t''$ be the vertices of $[y,x]$ and $[y,z]$ at
  distance $\lfloor (x|z)_y \rfloor$ from $y$. In order to prove the
  proposition, we need to show that $d(u,v),d(u',t'),d(v'',t'') \leq
  7\rho + 4$.

  We first show that $d(u',t') \leq 4 \rho +2$. Let $u'_z$ and
  $z_{u'}$ be the vertices of $[w,u']$ and $[w,z]$ at distance
  $\lfloor (u'|z)_w \rfloor$ from $w$. Let $t'_z$ and $z_{t'}$ be the
  vertices of $[w,t']$ and $[w,z]$ at distance $\lfloor (t'|z)_w
  \rfloor$ from $w$. Note that
  \begin{align*}
    d(u',t') &\leq d(u',u'_z) + d(u'_z,z_{u'}) +
    d(z_{u'},z_{t'})+ d(z_{t'},t'_z) + d(t'_z,t') \\ &
    \leq \lceil (w|z)_{u'} \rceil  + \rho + d(z_{u'},z_{t'})+ \rho +
     \lceil (w|z)_{t'} \rceil.
  \end{align*}
  Observe also that $d(z,z_{t'}) = d(z,t') - \lfloor (w|z)_{t'}
  \rfloor = \lfloor (x|y)_z \rfloor - \lfloor (w|z)_{t'} \rfloor$ and
  similarly, $d(z,z_{u'}) = \lfloor (x|y)_z \rfloor - \lfloor
  (w|z)_{u'} \rfloor$. Consequently, $d(z_{u'},z_{t'}) = |\lfloor
  (w|z)_{u'} \rfloor - \lfloor (w|z)_{t'} \rfloor|$. Therefore, we
  have
  $$ d(u',t') \leq 2\rho + 2 \max((w|z)_{u'},(w|z)_{t'}) + 1.$$

  Notice that $d(z,u')=d(z,t') \leq (x|y)_z = f-b \leq f +
  c= (w|y)_z$.  By Lemma~\ref{lem-rhotriangle}, $(w|z)_{u'},
  (w|z)_{t'} \leq \rho +\frac{1}{2}$, and consequently, $d(u',t') \leq
  4\rho +2$.

  We now show that $d(u,v) \leq 6 \rho + 3$. Note that if $b = 0$,
  then we are in the same case as for the pair $u',t'$, and thus we
  can assume that $b > 0$.  Let $u_x$ and $x_{u}$ be the vertices of
  $[w,u]$ and $[w,x]$ at distance $\lfloor (u|x)_w \rfloor$ from
  $w$. Let ${v}_x$ and $x_{v}$ be the vertices of $[w,v]$ and $[w,x]$
  at distance $\lfloor (v|x)_w \rfloor$ from $w$. Observe that
  \begin{align*}
    d(u,v) &\leq d(u,{u}_x) + d({u}_x,x_{u}) +
    d(x_{u},x_{v})+ d(x_{v},{v}_x) + d({v}_x,v) \\ &
    \leq \lceil (w|x)_{u} \rceil  + \rho + d(x_{u},x_{v})+ \rho +
     \lceil (w|x)_{v} \rceil.
  \end{align*}
  Observe also that $d(x,x_{u}) = d(x,u) - \lfloor (w|x)_{u}
  \rfloor = \lfloor (y|z)_x \rfloor - \lfloor (w|x)_{u} \rfloor$ and
  similarly, $d(x,x_{v}) = \lfloor (y|z)_x \rfloor - \lfloor
  (w|x)_{v} \rfloor$. Consequently, $d(x_{u},x_{v}) = |\lfloor
  (w|x)_{u} \rfloor - \lfloor (w|x)_{v} \rfloor|$. Therefore, we
  have
  $$ d(u,v) \leq 2\rho + 2 \max((w|x)_{u},(w|x)_{v}) + 1.$$

  Notice that $d(x,v) \leq (y|z)_x = e+b \leq e+b+c = (y|w)_x$ and
  that $d(x,u) = \lfloor e+b \rfloor = \lfloor (z|w)_x + b \rfloor
  \geq (z|w)_x$ (since $b > 0$). By Lemma~\ref{lem-rhotriangle},
  $(w|x)_v \leq \rho + \frac{1}{2}$, and $(w|x)_u = (w|z)_u +(w|x)_z
  -d(u,z) \leq (w|x)_z - d(u,z) + \rho +\frac{1}{2}$. Since $(w|x)_z =
  f$ and $d(u,z) = \lceil (x|y)_z \rceil = \lceil f-b \rceil$, we have
  $(w|x)_u \leq b +\rho +\frac{1}{2} \leq 2\rho +1$, and consequently,
  $d(u,v) \leq 6\rho +3$.

  We finally show that $d(v'',t'') \leq 7 \rho + 4$.  Note that if $c
  = 0$, then we are in the same case as for the pair $u,v$, and we can
  thus assume that $c >0$.  Let $v''_x$ and $x_{v''}$ be the vertices
  of $[w,v'']$ and $[w,x]$ at distance $\lfloor (v''|x)_w \rfloor$
  from $w$.  Let $t''_z$ and $z_{t''}$ be the vertices of $[w,t'']$
  and $[w,z]$ at distance $\lfloor (t''|z)_w \rfloor$ from $w$. Let
  $x_z$ and $z_x$ be the vertices of $[w,x]$ and $[w,z]$ at distance
  $\lfloor (x|z)_w \rfloor$ from $w$.

  Observe that
  \begin{align*}
    d(t'',v'') &\leq d(t'',t''_z) + d(t''_z,z_{t''}) +
    d(z_{t''},z_{x})+ d(z_{x},{x}_z) + d({x}_z,x_{v''}) +
    d(x_{v''},v''_x) + d(v''_x,v'') \\ &
    \leq \lceil (w|z)_{t''} \rceil  + \rho + d(z_{t''},z_{x})+ \rho +
      d({x}_z,x_{v''}) + \rho + \lceil (w|x)_{v''} \rceil.
  \end{align*}

  Notice that $d(z,z_{t''}) = d(z,t'') - \lfloor (w|z)_{t''} \rfloor
  = \lceil (x|y)_z \rceil - \lfloor (w|z)_{t''} \rfloor = \lceil f-b
  \rceil - \lfloor (w|z)_{t''} \rfloor $. Moreover, note that
  $d(z,z_x)= \lceil (x|w)_z \rceil = \lceil f \rceil$. Consequently,
  $d(z_x,z_{t''}) \leq \lceil b \rceil + \lfloor (w|z)_{t''} \rfloor$.

  Observe also that $d(x,x_{v''}) = d(x,v'') - \lfloor (x|w)_{v''}
  \rfloor = \lceil (y|z)_x \rceil - \lfloor (x|w)_{v''} \rfloor =
  \lceil e+b \rceil - \lfloor (x|w)_{v''}\rfloor$. Moreover, note that
  $d(x,x_z) = \lceil (z|w)_x \rceil = \lceil e \rceil$. Consequently,
  $d(x_z,x_{v''}) = |\lceil e+b \rceil - \lfloor (x|w)_{v''}\rfloor - \lceil e \rceil| \leq \max(\lceil b \rceil - \lfloor (x|w)_{v''}
  \rfloor, \lfloor (x|w)_{v''} \rfloor - \lfloor b
  \rfloor)$. Therefore, we have
  \begin{align*}
    d(t'',v'') &\leq 3\rho + \lceil (w|z)_{t''} \rceil + \lceil b
    \rceil + \lfloor (w|z)_{t''} \rfloor + d(x_z,x_{v''}) + \lceil (w|x)_{v''} \rceil \\
    &\leq 3\rho + 2 (w|z)_{t''}  + \lceil b
    \rceil +  \max(\lceil b \rceil -
    \lfloor (x|w)_{v''} \rfloor, \lfloor (x|w)_{v''} \rfloor - \lfloor
    b \rfloor) + \lceil (w|x)_{v''} \rceil \\
    & \leq 3 \rho + 2
    (w|z)_{t''} + 2 \max (\lceil b \rceil,(w|x)_{v''}) + 1.
  \end{align*}

  Notice that $d(x,v'')= \lceil (y|z)_x \rceil = \lceil e+b \rceil
  \leq e+b+c = (y|w)_x$ (since $c > 0$) and $d(z,t'') = \lceil (x|y)_z
  \rceil = \lceil f -b \rceil \leq f+c$ (since $c >0$). Recall that
  $f+c = (y|w)_z$. Consequently, by Lemma~\ref{lem-rhotriangle},
  $(w|x)_{v''},(w|z)_{t''} \leq \rho + \frac{1}{2}$. Since $b \leq
  \rho + \frac{1}{2}$, we get that $d(t'',v'') \leq 7 \rho + 4$.
\end{proof}

Consider a collection $\cT = (T_w)_{w \in V}$ of trees where for each
$w$, $T_w$ is an arbitrary BFS-tree rooted at $w$, and let $\rho_{\cT}(G)
:= \max_{w \in V} \rho_{w,T_w}(G)$.  Since for each $w$, $\rho_{w,T_w}(G)$
can be computed in $O(n^2)$ time, $\rho_{\cT}(G)$ can be computed in
$O(n^3)$ time.
We stress that for any fixed $w\in V$, $\delta_w(G)$ can be also
computed naively in $O(n^3)$ time and in $O(n^{2.69})$ time using
(max,min) matrix product~\cite{FouIsVi}. Furthermore, by
Proposition~\ref{hyp_basepoint}, $\delta_w(G)$ gives a 2-approximation
of the hyperbolicity $\delta(G)$ of $G$. In what follows, we present
approximation algorithms with similar running times for $\varsigma(G)$
and $\tau(G)$.

To get a better bound for $\varsigma(G)$, we need to involve one more
parameter.  Let $u$ and $v$ be arbitrary vertices of $G$ and $T_u \in
\cT$ be the BFS-tree rooted at $u$. Let also
$(u=u_0,u_1,\dots,u_{\ell}=v)$ be the path of $T_u$ joining $u$ with
$v$.  Define $\kappa_{T_u}(u,v):=\max\{d(a,u_i) : a\in I(u,v),
d(a,u)=i\}$ and $\kappa_{\cT}(G):=\max\{\kappa_{T_u}(u,v): u,v\in
V\}$. Note that $\kappa_{\cT}(G) \leq \kappa(G)$ and that $\kappa_{\cT}(G)$
can be computed in $O(n^3)$ time and $O(n^2)$ space.  Observe also
that for any $u,v$, $\kappa_{T_u}(u,v) \leq \rho_{u,T_u}(G)$ and thus
$\kappa_\cT(G) \leq \rho_\cT(G)$.

\begin{proposition} \label{main-th}
  For a graph $G$ and a collection of BFS-trees $\cT = (T_w)_{w\in
    V}$, $\iota(G) = \tau(G) \le \rho_\cT(G) + 2\kappa_\cT(G) \leq 3\rho_\cT(G)
  \le 3\tau(G)$ and $\varsigma(G)\le \rho_\cT(G)+2{\kappa_\cT(G)}\le
  8\varsigma(G)$.
  Consequently, a 3-approximation of the thinness $\tau(G)$ and an
  8-approximation of the slimness $\varsigma(G)$ can be found in
  $O(n^3)$ time and $O(n^2)$ space.
\end{proposition}

\begin{proof}
  Pick any geodesic triangle $\Delta(x,y,w)$ with sides $[x,y]$,
  $[x,w]$ and $[y,w]$.  Let $[x,w]_T$ and $[y,w]_T$ be the
  corresponding geodesics of the BFS-tree $T$ for vertex $w$. Consider
  the vertices $x_y\in [x,w]_T, y_x\in [w,y]_T$ and vertices $a\in
  [x,w], b\in [y,w]$ with
  $d(w,x_y)=d(w,y_x)=d(w,a)=d(w,b)=\lfloor(x|y)_w\rfloor$.  We know
  that $d(x_y,y_x)\le \rho_\cT(G)$. Since $(x|a)_w=d(a,w)$ and
  $(y|b)_w=d(b,w)$, $d(a,x_y)\le \kappa_{T_w}(w,x) \le \kappa_\cT(G)$ and
  $d(b,y_x)\le \kappa_{T_w}(w,y) \le \kappa_\cT(G)$.  Hence, $d(a,b)\le \rho_\cT(G) +
  2\kappa_\cT(G)$.
  Repeating this argument for vertices $x$ and $y$ and their
  BFS-trees, we get that the insize of $\Delta(x,y,w)$ is at most
  $\rho_\cT(G) + 2\kappa_\cT(G)$. So
  $\tau(G) \le \rho_\cT(G)+2\kappa_\cT(G)$ and by
  Proposition~\ref{prop:sl-vs-th},
  $\varsigma(G)\le \tau(G)\le \rho_\cT(G)+2\kappa_\cT(G)\le
  \tau(G)+2\kappa(G)\le 8\varsigma(G)$.
\end{proof}

\section{Exact computation} \label{exact} 

In this section, we provide exact algorithms for computing the
slimness $\varsigma(G)$, the thinness $\tau(G)$, and the insize
$\iota(G)$ of a given graph $G$. The algorithm computing
$\tau(G) = \iota(G)$ runs in $O(n^2m)$ time and the algorithm
computing $\varsigma(G)$ runs in $\widehat{O}(n^2m + n^4/\log^3 n)$
time (as we already noticed above, the $\widehat{O}(\cdot)$ notation
hides polyloglog factors); both algorithms are combinatorial and use
$O(n^2)$ space. When the graph is dense ({i.e.}, $m = \Omega(n^2)$),
that stays of the same order of magnitude as the best-known algorithms
for computing $\delta(G)$ in practice (see~\cite{Bo++}), but when the
graph is not so dense ({i.e.}, $m = o(n^2)$), our algorithms run in
$o(n^4)$ time.  In contrast to this result, the existing algorithms
for computing $\delta(G)$ exactly are not sensitive to the density of
the input.
We also show that the minsize $\rho_{-}(G)$ of a given graph $G$
cannot be approximated with a factor strictly better than 2 unless P = NP.
The main result of this section is the following theorem:

\begin{theorem} \label{thm:slim-thin-exact} For a graph $G=(V,E)$ with  $n$ vertices and $m$ edges, the following holds:
\begin{enumerate}[(1)]
\item the thinness $\tau(G)$ and the insize
$\iota(G)$ of $G$ can be computed in $O(n^2m)$ time;
\item the slimness $\varsigma(G)$ of $G$ can be computed in
  $\widehat{O}(n^2m + n^4/\log^3n)$ time combinatorially and in
  $O(n^{3.273})$ time using matrix multiplication;
\item deciding whether the minsize  $\rho_{-}(G)$ of $G$ is at most
  $1$ is NP-complete.
\end{enumerate}
\end{theorem}

One of the difficulties of computing 
$\varsigma(G),\tau(G),$ and $\iota(G)$ exactly is that these parameters 
are defined as minima of some functions over all geodesic triangles of 
the graph, and that there may be exponentially many such triangles. 
However, even in the case where there are unique shortest paths between 
all pairs of vertices, our algorithms have a better complexity than 
the naive algorithms following from the definitions of these parameters.

\subsection{Exact computation of thinness and insize}

In this subsection, we prove the following result
(Theorem~\ref{thm:slim-thin-exact}(1)):

\begin{proposition}\label{thm:thinness-exact}
$\tau(G)$ and  $\iota(G)$ can be computed in $O(n^2m)$ time.
\end{proposition}

To prove Proposition~\ref{thm:thinness-exact}, we introduce the
``pointed thinness'' $\tau_x(G)$ of a given vertex $x$. For a fixed
vertex $x$, let $\tau_x(G) = \max \big\{d(y',z') : \exists y,z \in V
\text{ such that } y' \in I(x,y), z' \in I(x,z), \text{ and } d(x,y')
= d(x,z') \leq (y|z)_x \big\}$. Observe that for any BFS-tree $T$
rooted at $x$, we have $\rho_{x,T}(G) \leq \tau_x(G) \leq \tau(G)$, and
thus by Corollary~\ref{cr:appr-th}, $\tau_x(G)$ is an 8-approximation
(with additive surplus 4) of $\tau(G)$. Since $\tau(G) = \max_{x \in
  V} \tau_x(G)$, given an algorithm for computing $\tau_x(G)$ in
$O(T(n,m))$ time, we can compute $\tau(G)$ in $O(nT(n,m))$ time, by
calling $n$ times this algorithm. Next, we describe such an algorithm
that runs in $O(nm)$ time for every $x$. By the remark above, the
latter will prove Theorem~\ref{thm:slim-thin-exact}(1).

Let $\tau_{x,y}(G) := \max \big\{d(y',z') : y' \in I(x,y)\text{ and }
\exists z \in V \text{ such that } z' \in I(x,z) \text{ and } d(x,y')
= d(x,z') \leq (y|z)_x \big\}$ and observe that $\tau_x(G) = \max_{y
  \in V} \tau_{x,y}(G)$.

For every ordered pair $x,y$ and every vertex $w$, let $g_w(x,y) =
\max \big\{d(y',w) : y' \in I(x,y) \text{ and } d(x,y') =
d(x,w)\big\}$ and let $h_{x,y}(w) = \max \big\{(y|z)_x : w \in
I(x,z)\big\}$. The following lemma is the cornerstone of our
algorithm.

\begin{lemma}\label{lem:rec-thin}
  For any $x,y \in V$, $\tau_{x,y}(G) = \max\big\{g_w(x,y) : d(x,w) \leq
  h_{x,y}(w)\big\}$.
\end{lemma}

\begin{proof}
Let $\beta_{x,y} := \max\big\{g_w(x,y) : d(x,w) \leq h_{x,y}(w)\big\}$
and consider $w$ such that $\beta_{x,y} = g_w(x,y)$ and $d(x,w) \leq
h_{x,y}(w)$. Consider a vertex $y' \in I(x,y)$ such that $d(x,y') =
d(x,w)$ and $d(y',w) = g_w(x,y)$. Consider a vertex $z$ such that $w
\in I(x,z)$ and $h_{x,y}(w) = (y|z)_x$. Since $d(x,y') = d(x,w) \leq
h_{x,y}(w) = (y|z)_x$,  $\beta_{x,y} = g_w(x,y) = d(y',w) \leq
\tau_{x,y}(G)$.

Conversely, consider $y',z',z$ such that $y' \in I(x,y)$, $z' \in
I(x,z)$, $d(x, y') = d(x,z') \leq (y|z)_x$, and $\tau_{x,y}(G) =
d(y',z')$. Observe that $d(y',z') \leq g_{z'}(x,y)$ and that $d(x,z')
\leq (y|z)_x \leq h_{x,y}(z')$. Consequently, $\tau_{x,y}(G) = d(y',z')
\leq g_{z'}(x,y) \leq \beta_{x,y}$.
\end{proof}

The algorithm for computing $\tau_x(G)$ works as follows. First, we
compute the distance matrix of $G$ in $O(mn)$ time. Next, we compute
$g_w(x,y)$ and $h_{x,y}(w)$ for all $y,w$ in time $O(mn)$. Finally, we
enumerate all $y,w$ in $O(n^2)$ to compute $\max\big\{g_w(x,y) :
d(x,w) \leq h_{x,y}(w)\big\}$. By Lemma~\ref{lem:rec-thin}, the
obtained value is exactly $\tau_x(G) = \max \tau_{x,y}(G)$. Therefore,
we are just left with proving that we can compute $g_w(x,y)$ and
$h_{x,y}(w)$ for all $y,w$ in time $O(mn)$, which is a direct
consequence of the two next lemmas.

\begin{lemma}\label{lem:compute-thinness-projection}
For any fixed $x,w \in V$, one can compute the values of $g_w(x,y)$
for all $y \in V$ in ${O}(m)$ time.
\end{lemma}

\begin{proof}
In order to compute $g_w(x,y)$, we use the following recursive
formula: $g_w(x,y) = 0$ if $d(x,y) < d(x,w)$, $g_w(x,y) = d(w,y)$ if
$d(x,y) = d(x,w)$, and $g_w(x,y) = \max \big\{ g_w(x,y'): y' \in N(y)
\text{ and } d(x,y') = d(x,y)-1\big\}$ otherwise.  Given the distance
matrix $D$, for any $y \in V$, we can compute $\{y' \in N(y) : d(x,y')
= d(x,y)-1\}$ in ${O}(\deg(y))$ time.  Therefore, using a standard
dynamic programming approach, we can compute the values $g_w(x,y)$ for
all $y \in V$ in ${O}(\sum_y \deg(y)) = O(m)$ time.
\end{proof}

\begin{lemma}\label{lem:compute-thinness-descendants}
For any fixed $x,y \in V$, one can compute the values of $h_{x,y}(w)$
for all $w \in V$ in ${O}(m)$ time.
\end{lemma}

\begin{proof}
In order to compute $h_{x,y}(w)$, we use the following recursive
formula: $h_{x,y}(w) = \max\big\{(y|w)_x, h'_{x,y}(w)\big\}$ where
$h'_{x,y}(w) = \max \big\{h_{x,y}(w') : w' \in N(w) \text{ and }
d(x,w') = d(x,w)+1 \big\}$.  Given the distance matrix $D$, for any
fixed $w \in V$, we can compute $\big\{w' \in N(w) : d(x,w') =
d(x,w)+1 \big\}$ in $O(\deg(w))$ time.  If we order the vertices of $V$
by non-increasing distance to $x$, using dynamic programming, we can
compute the values of $h_{x,y}(w)$ for all $w$ in $O(\sum_w \deg(w)) =
O(m)$ time.
\end{proof}

\subsection{Exact computation of slimness}

The goal of this subsection is to prove the following result  (Theorem~\ref{thm:slim-thin-exact}(2)):

\begin{proposition}\label{thm:slimness-exact}
  $\varsigma(G)$ can be computed in $\widehat{O}(n^2m + n^4/\log^3n)$
  time combinatorially and in $O(n^{3.273})$ time using matrix
  multiplication.
\end{proposition}

To prove Proposition \ref{thm:slimness-exact}, we introduce the
``pointed slimness'' $\varsigma_w(G)$ of a given vertex $w$.
Formally, $\varsigma_w(G)$ is the least integer $k$ such that, in any
geodesic triangle $\Delta(x,y,z)$ such that $w \in [x,y]$, we have
$d(w,[x,z] \cup [y,z]) \leq k$. Note that $\varsigma_w(G)$ cannot be
used to approximate $\varsigma(G)$ (that is in sharp contrast with
$\delta_w(G)$ and $\tau_w(G)$).  In particular, $\varsigma_w(G) = 0$
whenever $w$ is a {\it pending vertex} (a vertex of degree 1), or,
more generally, a {\it simplicial vertex} (a vertex whose every two
neighbors are adjacent) of $G$.  On the other hand, we have
$\varsigma(G) = \max_{w \in V} \varsigma_w(G)$.  Therefore, given an
algorithm for computing $\varsigma_w(G)$ in $O(T(n,m))$ time, we can
compute $\varsigma(G)$ in $O(nT(n,m))$ time, by calling $n$ times this
algorithm.  Next we describe such an algorithm that is combinatorial
and runs in $\widehat{O}(nm + n^3/\log^3n)$
(Lemma~\ref{lem:compute-pointed-slimness}). We also explain how to
compute $\varsigma(G)$ in $O(n^{2.373})$ time using matrix
multiplication (Corollary~\ref{cor:compute-pointed-slimness}).  By the
remark above, it will prove Theorem~\ref{thm:slim-thin-exact}(2).  For
every $y,z \in V$ we set $p_w(y,z)$ to be the least integer $k$ such
that, for every geodesic $[y,z]$, we have $d(w,[y,z]) \leq k$.  The
following lemma is the cornerstone of our algorithm.

\begin{lemma}\label{lem:pointed-slimness}
$\varsigma_w(G) \leq k$ iff for all $x,y \in V$ such that $w \in I(x,y)$, and any $z \in V$,  $\min \{p_w(x,z),p_w(y,z)\} \leq k$.
\end{lemma}

\begin{proof}
  In one direction, let $\Delta(x,y,z)$ be any geodesic triangle such
  that $w \in [x,y]$.  Then,
  $d(w,[x,z]\cup[y,z]) \leq \min \{p_w(x,z),p_w(y,z)\} \leq k$.  Since
  $\Delta(x,y,z)$ is arbitrary, $\varsigma_w(G) \leq k$. Conversely,
  assume that $\varsigma_w(G) \leq k$.  Let $x,y,z \in V$ be arbitrary
  vertices such that $w \in I(x,y)$.  Consider a geodesic triangle
  $\Delta(x,y,z)$ by selecting its sides in such a way that
  $w\in [x,y]$ and $d(w,[x,z]) = p_w(x,z), d(w,[y,z]) = p_w(y,z)$
  hold.  Then
  $d(w,[x,z]\cup[y,z]) = \min \{p_w(x,z),p_w(y,z)\} \leq
  \varsigma_w(G) \leq k$, and we are done.
\end{proof}

The algorithm for computing $\varsigma_w(G)$ proceeds in two phases.
We first compute $p_w(y,z)$ for every $y,z \in V$.  Second, we seek
for a triplet $(x,y,z)$ of distinct vertices such that
$w \in I(x,y)$ and $\min\{p_w(x,z),p_w(y,z)\}$ is maximized.  By
Lemma~\ref{lem:pointed-slimness}, the obtained value is exactly
$\varsigma_w(G)$.  Therefore, we are just left with proving the
running time of our algorithm.

\begin{lemma}\label{lem:compute-slimness-projection}
The values $p_w(y,z)$, for all $y,z \in V$, can be computed in ${O}(nm)$ time.
\end{lemma}

\begin{proof}
By induction on $d(y,z)$, the following formula holds for $p_w(y,z)$:
$p_w(y,z) = d(w,y)$ if $y=z$; otherwise, $p_w(y,z) = \min \{ d(w,y), \ \max \{ p_w(x,z): x \in N(y) \cap I(y,z)\} \}$.
Since the distance matrix $D$ of $G$ is available, for any $y,z \in V$ and
for any $x \in N(y)$, we can check in constant time whether $x \in I(y,z)$ ({i.e.}, whether $d(x,z) = d(y,z) - 1$).
In particular, given $y \in V$, for every of the $n$ possible choices for $z$, the intersection $N(y) \cap I(y,z)$ can be computed in ${O}(\deg(y))$ time.
Therefore, using a standard dynamic programming approach, all the values $p_w(y,z)$ can be computed in time ${O}(nm+ \sum_y n\cdot \deg(y))$,
that is in ${O}(nm)$.
\end{proof}

We note that once the distance-matrix of $G$ has been precomputed, and
we have all the values $p_w(y,z)$, for all $y,z \in V$, then we can
compute $\varsigma_w(G)$ as follows.  We enumerate all possible
 triplets $(x,y,z)$ of distinct vertices of $G$, and we keep one
such that $w \in I(x,y)$ and $\min\{p_w(x,z),p_w(y,z)\}$ is maximized.
It takes $O(n^3)$ time.  In what follows, we shall explain how the
running time can be improved by reducing the problem to {\sc Triangle
  Detection}.
More precisely, let $k$ be a fixed integer.
The graph $\Gamma_{\varsigma}[k]$ has vertex set $V_1 \cup V_2 \cup V_3$, with every set $V_i$ being a copy of $V \setminus \{ w\}$.
There is an edge between $x_1 \in V_1$ and $y_2 \in V_2$ if and only if the corresponding vertices $x,y \in V$ satisfy $w \in I(x,y)$.
Furthermore, there is an edge between $x_1 \in V_1$ and $z_3 \in V_3$ (respectively, between $y_2 \in V_2$ and $z_3 \in V_3$) if and only if we have $p_w(x,z) > k$ (respectively, $p_w(y,z) > k$).

\begin{lemma}\label{lem:reduction-slimness-to-triangle}
$\varsigma_w(G) \leq k$ if and only if $\Gamma_{\varsigma}[k]$ is triangle-free.
\end{lemma}

\begin{proof}
  By construction there is a bijective correspondence between the
  triangles $(x_1,y_2,z_3)$ in $\Gamma_{\varsigma}[k]$ and the triplets
  $(x,y,z)$ such that $w \in I(x,y)$ and
  $\min \{p_w(x,z),p_w(y,z)\} > k$.  By
  Lemma~\ref{lem:pointed-slimness}, we have $\varsigma_w(G) \leq k$ if
  and only if there is no triplet $(x,y,z)$ such that $w \in I(x,y)$
  and $\min \{p_w(x,z),p_w(y,z)\} > k$.  As a result,
  $\varsigma_w(G) \leq k$ if and only if $\Gamma_{\varsigma}[k]$ is
  triangle-free.
\end{proof}

\begin{lemma}\label{lem:compute-pointed-slimness}
  For $w \in V$, we can compute $\varsigma_w(G)$ in
  $\widehat{O}(nm + n^3/\log^3n)$ time combinatorially.
\end{lemma}

\begin{proof}
  We compute the values $p_w(y,z)$, for every $y,z \in V$.  By
  Lemma~\ref{lem:compute-slimness-projection}, it takes time
  ${O}(nm)$.  Furthermore, within the same amount of time, we can also
  compute the distance matrix $D$ of $G$.  Then, we need to observe
  that given an algorithm to decide whether $\varsigma_w(G) \leq k$
  for any $k$, that runs in $O(T(n,m))$ time, we can compute
  $\varsigma_w(G)$ in $O(T(n,m)\log n)$ time, simply by performing a
  one-sided binary search.  In what follows, we describe such an
  algorithm that runs in time $\widehat{O}(n^3/\log^4n)$.  For that,
  we reduce the problem to {\sc Triangle Detection}.  We construct the
  graph $\Gamma_{\varsigma}[k]$.  Since the values $p_w(y,z)$, for all
  $y,z \in V$, and the distance matrix of $G$ are given, this can be
  done in ${O}(n^2)$ time.  Furthermore, by
  Lemma~\ref{lem:reduction-slimness-to-triangle},
  $\varsigma_w(G) \leq k$ if and only if $\Gamma_{\varsigma}[k]$ is
  triangle-free.  Since {\sc Triangle Detection} can be solved
  combinatorially in time $\widehat{O}(n^3/\log^4n)$ \cite{Yu2015}, we
  are done by calling ${O}(\log{n})$ times a {\sc Triangle Detection}
  algorithm.
\end{proof}

Interestingly, in the proof of Lemma~\ref{lem:compute-pointed-slimness} we reduced the computation of $\varsigma_w(G)$ to a single call to an all-pair-shortest-path algorithm, and to ${O}(\log{n})$ calls to a {\sc Triangle Detection} algorithm. It is folklore that both problems can be solved in time ${O}(n^{\omega}\log{n})$ and ${O}(n^{\omega})$, respectively, where $\omega < 2.373$ is the exponent for square matrix multiplication. Hence, we obtain the following algebraic version of Lemma~\ref{lem:compute-pointed-slimness}:

\begin{corollary}\label{cor:compute-pointed-slimness}
  For $w \in V$, we can compute $\varsigma_w(G)$ in
  ${O}(n^{\omega}\log{n})$ time.
\end{corollary}

We stress that Corollary~\ref{cor:compute-pointed-slimness} implies
the existence of an ${O}(n^{\omega+1}\log{n})$-time algorithm for
computing the slimness of a graph (since $\omega < 2.373$, this
algorithm runs in ${O}(n^{3.273})$ time). In sharp contrast to this
result, we recall that the best-known algorithm for computing the
hyperbolicity runs in time $O(n^{3.69})$~\cite{FouIsVi}.

A popular conjecture is that {\sc Triangle Detection} and {\sc Matrix Multiplication} are equivalent. We prove next that under this assumption, the result of Corollary~\ref{cor:compute-pointed-slimness} is optimal up to polylogarithmic factors:

\begin{proposition}\label{reduction-triangle-pointed-slimness} {\sc Triangle
    Detection} on $n$-vertex graphs can be reduced in time ${O}(n^2)$
  to computing the pointed slimness of a given vertex in a graph with
  $\Theta(n)$-vertices.
\end{proposition}

\begin{proof}
  Let $G=(V,E)$ be any graph input for {\sc Triangle
    Detection}. Suppose without loss of generality that $G$ is
  tripartite with a valid partition $V_1,V_2,V_3$ (otherwise, we
  replace $G$ with $H=(V_1\cup V_2\cup V_3,E_H)$ where $V_1,V_2,V_3$
  are disjoint copies of $V$ and
  $E_H = \{ x_iy_j: xy \in E \ \text{and} \ 1 \leq i < j \leq 3
  \}$). We construct a graph $G^*$ from $G$, as follows.
  \begin{itemize}
  \item For every $v \in V$, there is a path $(v^-,v^*,v^+)$.  We so
    have three copies of the partition sets $V_i$, $1 \leq i \leq 3$,
    that we denote by $V_i^-, V_i^*, V_i^+$.
  \item For every $xz \in E$ such that $x \in V_1, z \in V_3$, we add
    an edge $x^-z^+$.  In the same way, for every $yz \in E$ such that
    $y \in V_2, z \in V_3$, we add an edge $y^+z^-$.  However, for
    every $x \in V_1, y \in V_2$ we add an edge $x^+y^-$ if and only
    if $xy \notin E$.
  \item We also add two new vertices $\alpha,\beta$ and the edges
    $\{ \alpha x^* : x \in V_1 \} \cup \{ \beta y^* : y \in V_2
    \}$.
  \item Finally, we add two more vertices $a,b$ and the edges
    $\{ ab, a\alpha, a\beta \} \cup \{ ax^+, bx^- : x \in V_1 \}
    \cup \{ ay^-, by^+ : y \in V_2 \} \cup \{ bz^-, bz^+ : z \in
    V_3 \}$.
\end{itemize}
The resulting graph $G^*$ has ${O}(n)$ vertices and it can be
constructed in ${O}(n^2)$-time (for an illustration, see Fig. \ref{fig_G*}).  
In what follows, we prove
that $\varsigma_a(G^*) \geq 2$ if and only if $G$ contains a triangle.

First we assume that $G$ contains a triangle $xyz$ where
$x \in V_1, y \in V_2, z \in V_3$. By construction, the paths
$(x^*,x^-,z^+,z^*)$ and $(z^*,z^-,y^+,y^*)$ are geodesics and they do
not intersect $N_{G^*}[a]$. Furthermore, since $xy \in E$, we cannot
find any two neighbors of $x^*$ and $y^*$, respectively, that are
adjacent, thereby implying $d_{G^*}(x^*,y^*) = 4$ ({e.g.},
$(x^*,x^+,a,y^-,y^*)$ is a geodesic). Overall, the triplet
$x^*,y^*,z^*$ is such that $a \in I(x^*,y^*)$,
$p_a(x^*,z^*) = p_a(y^*,z^*) = 2$. As a result,
$\varsigma_a(G^*) \geq 2$.

Conversely, assume $\varsigma_a(G^*) \geq 2$.  Let $r,s,t \in V(G^*)$
such that: $a \in I(r,s)$, $p_a(r,t) \geq 2$ and in the same way
$p_a(s,t) \geq 2$.  We claim that $r = x^*$ for some $x \in V$.
Indeed, suppose by way of contradiction that this is not the case.  By the
hypothesis $r \notin N_{G^*}[a]$, and so, $r \in \{v^+,v^-\}$ for some
$v \in V$ and $r \in N_{G^*}(b)$.  Furthermore, $d_{G^*}(r,a) = 2$,
and so, since $a \in I(r,s)$ and by the hypothesis
$s \notin N_{G^*}[a]$, $d_{G^*}(r,s) \geq 4$.  However, by
construction every vertex of $G^*$ is at a distance $\leq 2$ from
vertex $b$.  Since $b \in N_{G^*}(r)$, this implies that $r$ has
eccentricity at most three, a contradiction.  Therefore, we proved as
claimed that $r = x^*$ for some $x \in V$.  We can prove similarly
that $s = y^*$ for some $y \in V$.  Then, observe that we cannot have
$x,y \in V_1$ (otherwise, $(x^*,\alpha,y^*)$ is a geodesic,
$d_{G^*}(r,s) =d_{G^*}(x^*,y^*)= 2$ and $a \notin I(r,s)$); we cannot
have $x,y \in V_2$ either.  Finally, we cannot have $x \in V_3$ for
then we would get
$d_{G^*}(x^*,a) + d_{G^*}(a,y^*) \geq 3 + 2 = 5 > 4 \geq
d_{G^*}(x^*,y^*)$; for the same reason, we cannot have $y \in V_3$.
Overall, we may assume without loss of generality that
$x \in V_1, y \in V_2$.  Note that $xy \notin E$ (otherwise,
$d_{G^*}(x^*,y^*) = 3$ and $a \notin I(x^*,y^*)$).  Let $P$ be a
shortest $(x^*,t)$-path and $Q$ be a shortest $(y^*,t)$-path such that
$d_{G^*}(a,P) \geq 2, \ d_{G^*}(a,Q) \geq 2$.  Since
$V_3^* \cup N_{G^*}[a]$ intersects any path from $V_1^*$ to $V_2^*$,
there exists  $z \in V_3$ such that $z^* \in P \cup Q$.  By symmetry,
we may assume $z^* \in P$.  It follows from the construction of $G^*$
that the unique shortest $(x^*,z^*)$-path that does not intersect
$N_{G^*}[a]$, if any, must be $(x^*,x^-,z^+,z^*)$.  In particular,
$xz \in E$.  Suppose by contradiction $t \neq z^*$. Then,
$(x^*,x^-,z^+,z^*,z^-)$ is a subpath of $P$, that is impossible 
because $d_{G^*}(x^*,z^-) = 3$.  Therefore, $t = z^*$.  We prove
similarly as before that $yz \in E$.  Summarizing, $xyz$ is a triangle
of $G$.
\end{proof}

\begin{figure}
\begin{center}
\includegraphics[scale=0.5]{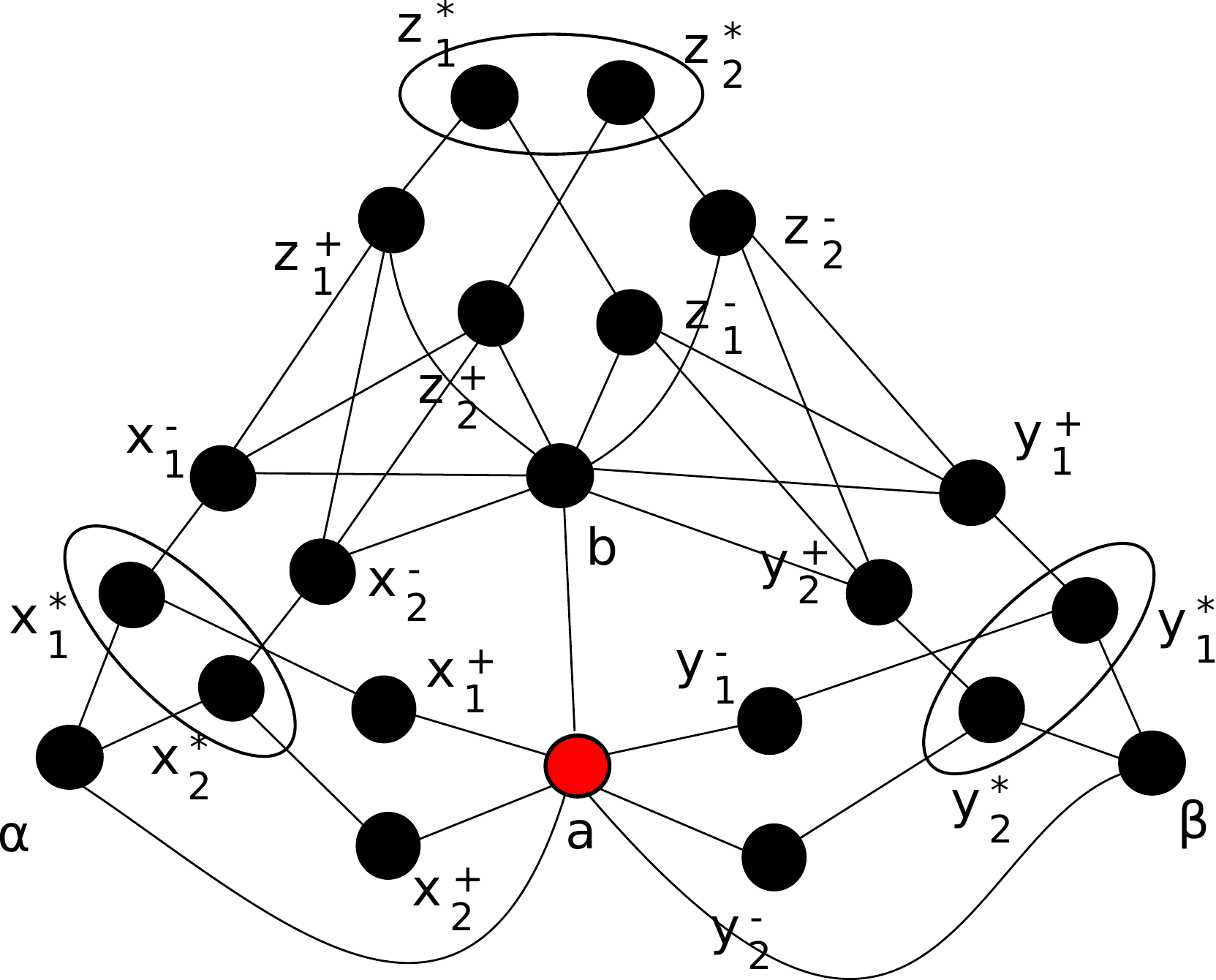}
\end{center}
\caption{The graph $G^*$ obtained from a tripartite graph $G$ and used in the proof of Proposition \ref{reduction-triangle-pointed-slimness}.} \label{fig_G*}
\end{figure}

\subsection{Approximating the minsize is hard} \label{minsize}

In this subsection we prove that, at the difference from other
hyperbolicity parameters, deciding whether $\rho_{-}(G) \leq 1$ is
NP-complete (Theorem~\ref{thm:slim-thin-exact}(3)). Note that since
$\rho_{-}(G)$ is an integer, this immediately implies that we cannot
find a $(2-\epsilon)$-approximation algorithm to compute $\rho_{-}(G)$
unless P = NP.

\begin{proposition}\label{thm:minsize}
Deciding if $\rho_{-}(G)\le 1$ is NP-complete.
\end{proposition}

Note that if we are given a BFS-tree $T$ rooted at $w$, we can easily
check whether $\rho_{w,T}(G) \leq 1$, and thus deciding whether
$\rho_-(G) \leq 1$ is in NP.  In order to prove that this problem is
NP-hard, we do a reduction from \textsc{Sat}.

Let $\Phi$ be a \textsc{Sat} formula with $m$ clauses $c_1,c_2,
\ldots,c_m$ and $n$ variables $x_1,x_2,\ldots,x_n$. Up to
preprocessing the formula, we can suppose that $\Phi$ satisfies the
following properties (otherwise, $\Phi$ can be reduced to a formula
satisfying these conditions):
\begin{itemize}
\item no clause $c_j$ can be reduced to a singleton; 
\item every literal $x_i,\overline{x}_i$ is contained in at least one clause; 
\item no clause $c_j$ can contain both $x_i,\overline{x}_i$; 
\item no clause $c_j$ can be strictly contained in another clause $c_k$; 
\item every clause $c_j$ is disjoint from some other clause $c_k$ (otherwise, a trivial satisfiability assignment for $\Phi$ is to set true every literal in $c_j$);
\item if two clauses $c_j,c_k$ are disjoint, then there exists another clause $c_p$ that intersects $c_j$ in {\em exactly one} literal, and similarly, that also intersects $c_k$ in {\em exactly one} literal (otherwise, we add the two new clauses $x \vee \overline{y}$ and $\overline{x} \vee y$, with $x,y$ being fresh new variables; then, we replace every clause $c_j$ by the two new clauses $c_j \vee x \vee y$ and $c_j \vee \overline{x} \vee \overline{y}$).
\end{itemize}

Let $X := \{x_1,\overline{x}_1, \ldots, x_{n},\overline{x}_n\}$.
For simplicity, in what follows, we  often denote $x_i,\overline{x}_i$ by ${\ell}_{2i-1},{\ell}_{2i}$.
Let $C := \{c_1,\ldots,c_m\}$ be the clause-set of $\Phi$.
Finally, let $w$ and $V = \{v_1,v_2,\ldots,v_{2n}\}$ be additional vertices.
We construct a graph $G_{\Phi}$ with $V(G_\Phi)
=\{w\} \cup V \cup X \cup C$ and where $E(G_{\Phi})$ is defined as
follows:
\begin{itemize}
\item $N(w) = V$ and $V$ is a clique,
\item for every $i,i'$, $v_i$ and ${\ell}_{i'}$ are adjacent if and only if
  $i=i'$;
\item for every $i,i'$, ${\ell}_i$ and ${\ell}_{i'}$ are adjacent if and only if
  ${\ell}_{i'} \neq \overline{\ell}_i$;
\item for every $i,j$, $v_i$ and $c_j$ are not adjacent;
\item for every $i,j$, ${\ell}_i$ and $c_j$ are adjacent if and only if
  ${\ell}_i \in c_j$;
\item for every $j,j'$, $c_j,c_{j'}$ are adjacent if and only if $c_j,c_{j'}$
   intersect in exactly one literal.
\end{itemize}

\begin{figure}[!ht]
\centering
\begin{subfigure}[b]{.45\textwidth}\centering
\includegraphics[width=.75\textwidth]{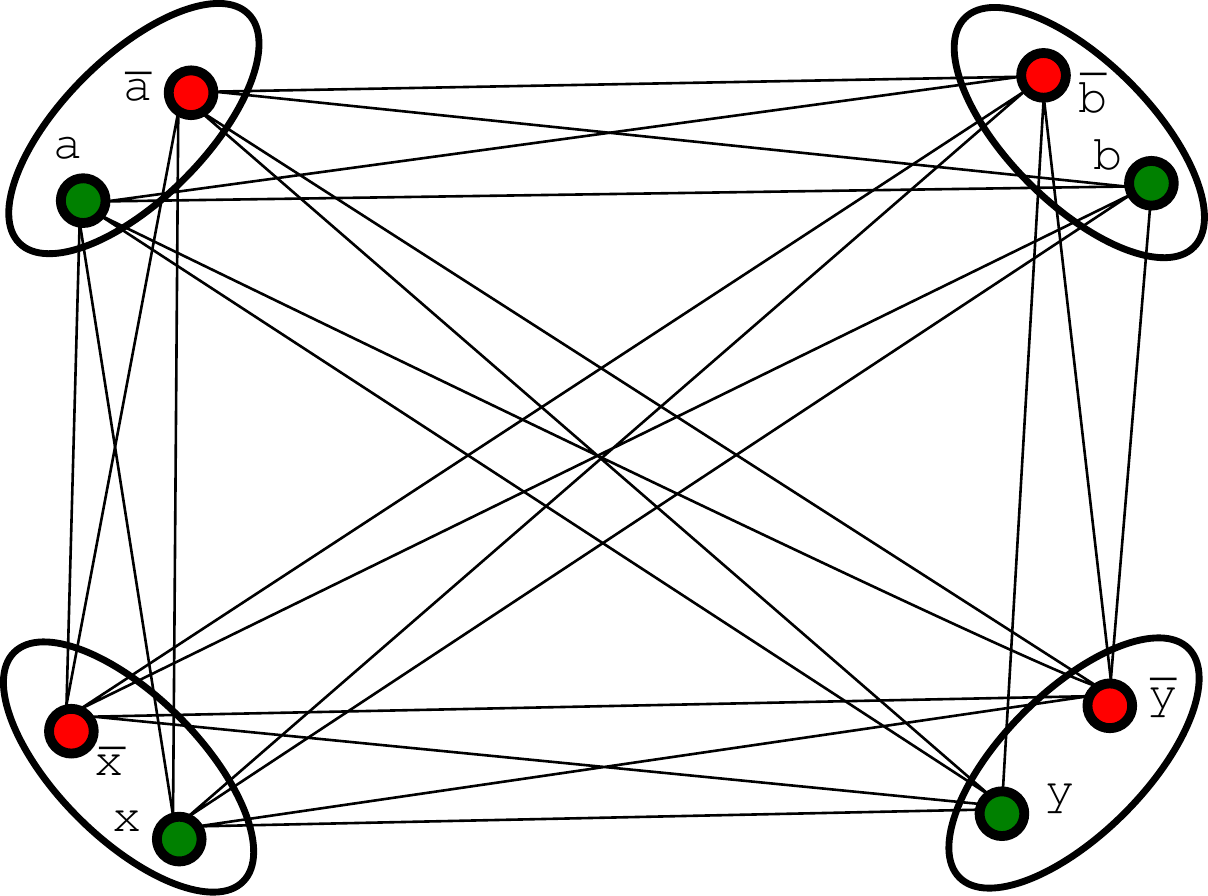}
\caption{Edges between literal vertices.}
\label{fig:literal}
\end{subfigure}
\hfill
\begin{subfigure}[b]{.45\textwidth}\centering
\includegraphics[width=.75\textwidth]{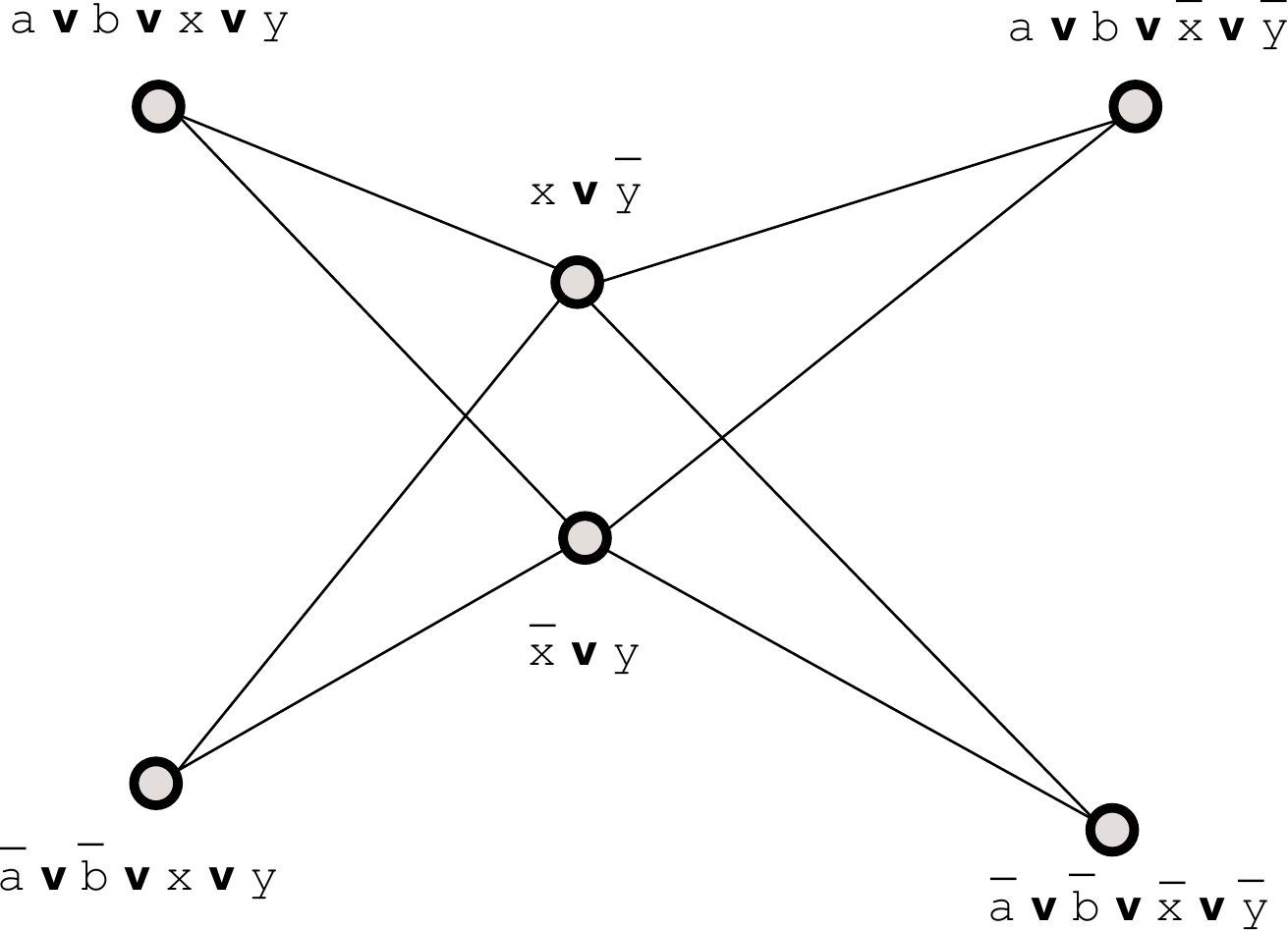}
\caption{Edges between clause vertices.}
\label{fig:clause}
\end{subfigure}
\vfill
\begin{subfigure}[b]{.9\textwidth}\centering
\includegraphics[width=.75\textwidth]{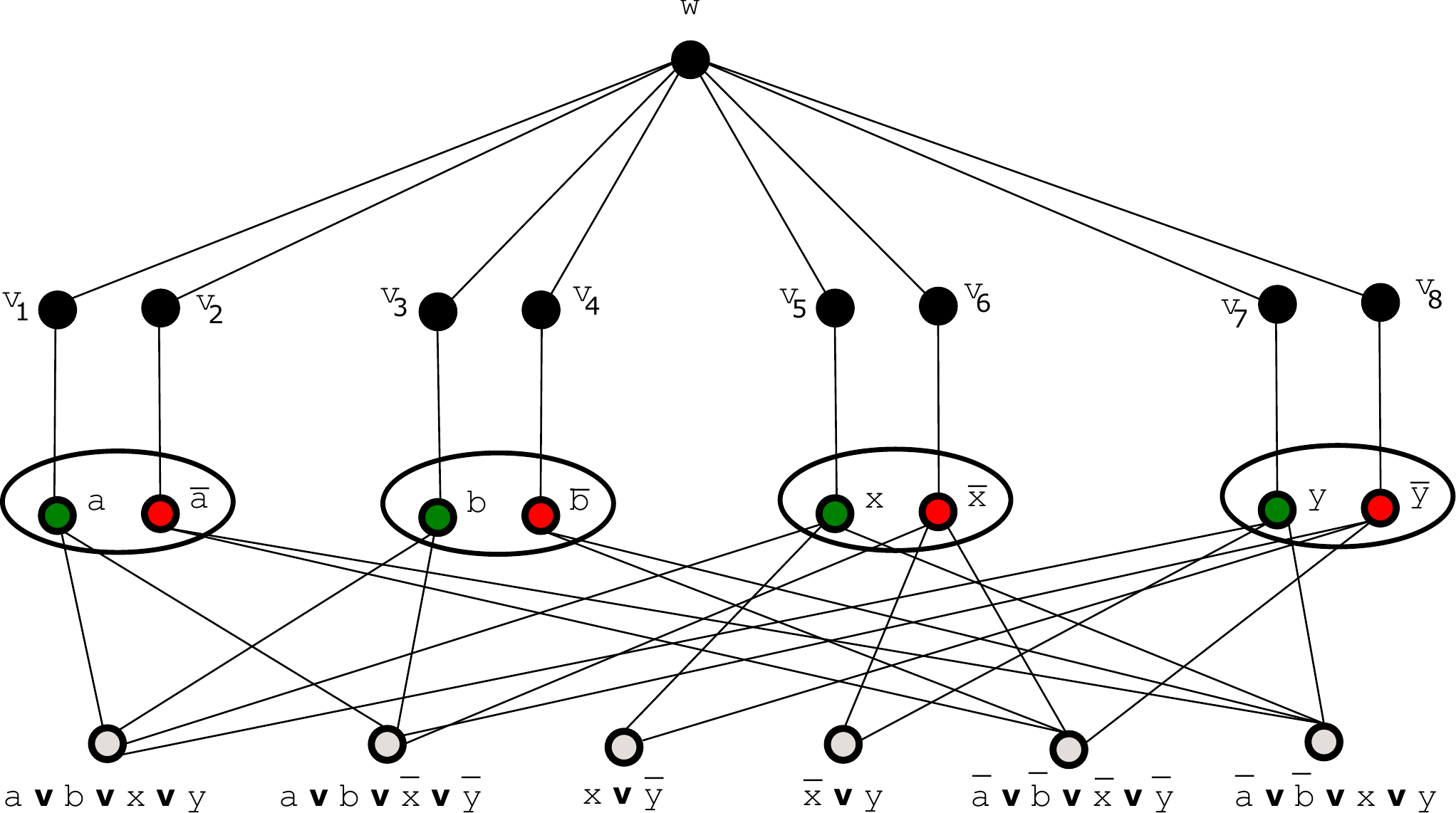}
\caption{Edges between consecutive layers.}
\label{fig:global}
\end{subfigure}
\vfill
\begin{subfigure}[b]{.9\textwidth}\centering
\includegraphics[width=.75\textwidth]{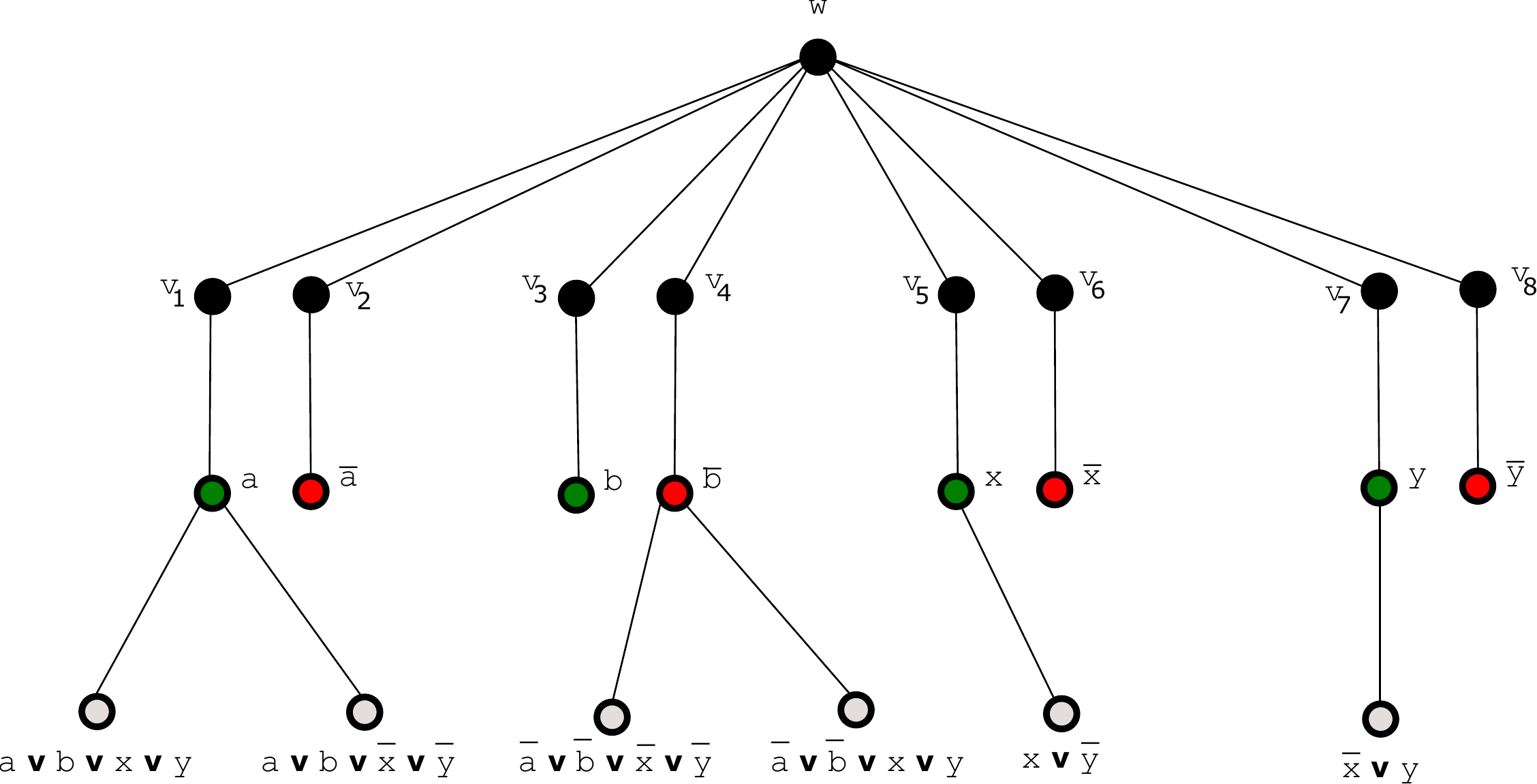}
\caption{A BFS tree $T$ rooted at $w$ s.t. $\rho_{w,T}(G_{\Phi}) = 1$.}
\label{fig:solution}
\end{subfigure}
\caption{The graph $G_{\Phi}$ obtained from the formula $\Phi = (a \vee b) \land (\bar{a} \vee \bar{b})$. After preprocessing $\Phi$, we got the equivalent formula $(a \vee b \vee x \vee y) \land (a \vee b \vee \bar{x} \vee \bar{y}) \land (\bar{a} \vee \bar{b} \vee x \vee y) \land (\bar{a} \vee \bar{b} \vee \bar{x} \vee \bar{y}) \land (\bar{x} \vee y) \land (x \vee \bar{y})$.}
\label{fig:reduction}
\end{figure}
We refer to Fig.~\ref{fig:reduction} for an illustration.

\begin{proposition} \label{np-minsize}
${\rho}_{-}(G_{\Phi}) \leq 1$ if and only if $\Phi$ is satisfiable.
\end{proposition}

In order to prove the hardness result, we start by showing that in order to get $\rho_{r,T}(G_{\Phi}) \leq 1$ we must have $r = w$.
In the following proofs, by a {\it parent node} of a node we mean its parent in  a BFS-tree $T$.

\begin{lemma}
For every BFS-tree $T$ rooted at $c_j\in C$, we have $\rho_{c_j,T}(G_{\Phi}) \geq 2$.
\end{lemma}

\begin{proof}
Suppose for the sake of contradiction that we have $\rho_{c_j,T}(G_{\Phi}) \leq 1$.
Let $X_j \subseteq X$ be the literals in $c_j$.
Note that since $|X_j| > 1$, every vertex in $\overline{X}_j$ is at distance two from $c_j$.
We claim that for every ${\ell}_{i'} \in \overline{X}_j$, the parent node of ${\ell}_{i'}$  is in $X_j$.
Indeed, otherwise this would be some clause-vertex $c_k$ such that ${\ell}_{i'} \in c_k$ and $c_k \cap c_j \neq \emptyset$.
Then, let ${\ell}_i \in c_j \setminus c_k$.
Vertex ${\ell}_i$ must be the parent node of $v_i$.
However, $(v_i|{\ell}_{i'})_{c_j} = (2+2-2)/2 = 1$.
The latter implies that $c_k,{\ell}_i$ should be adjacent, that contradicts the fact that ${\ell}_i \notin c_k$.
Therefore, the claim is proved.

Then, let $c_k$ be disjoint from $c_j$.
By construction, $d(c_j,c_k) = 2$, and the parent node of $c_k$  must be some $c_p$ such
that $c_j \cap c_p \neq \emptyset$, and similarly $c_p \cap c_k \neq \emptyset$.
Let ${\ell}_{i'} \in c_j \cap c_p$.
Furthermore, let ${\ell}_i \in c_j$ be the parent node of its {\em negation} $\overline{\ell}_{i'}$ in $T$.
We stress that ${\ell}_i \notin c_p$ since ${\ell}_{i'}$ is the unique literal contained in $c_j \cap c_p$.
We have $(c_k|\overline{\ell}_{i'})_{c_j} = 2 - d(c_k,\overline{\ell}_{i'})/2 \in \{1,3/2\}$.
In particular, $\lfloor (c_k|\overline{\ell}_{i'})_{c_j} \rfloor = 1$.
As a result, $\rho_{c_j,T}(G_{\Phi}) \geq d(c_p,{\ell}_{i}) = 2$.
\end{proof}

\begin{lemma}
For every BFS-tree $T$ rooted at ${\ell}_i\in X$, we have $\rho_{{\ell}_i,T}(G_{\Phi}) \geq 2$.
\end{lemma}

\begin{proof}
Suppose for the sake of contradiction that $\rho_{{\ell}_i,T}(G_{\Phi}) \leq 1$.
Since there is a perfect matching between $X$ and $V = N(w)$, the parent node of $w$ must be $v_i$.
We claim that the parent node of $v_{i'}$, for every $i' \neq i$, must be also $v_i$.
Indeed, otherwise this should be ${\ell}_{i'}$.
However, $(w|v_{i'})_{{\ell}_i} = (2+2-1)/2 = 3/2$.
In particular, $\lfloor (w|v_{i'})_{{\ell}_i} \rfloor = 1$, and so, $\rho_{{\ell}_i,T}(G_{\Phi}) \leq 1$
implies $v_i$ and ${\ell}_{i'}$ should be adjacent, that is a contradiction.
So, the claim is proved.

Then, let $c_j \in C$ be nonadjacent to ${\ell}_i$.
We have $d(c_j,{\ell}_i) = 2$, and the parent node $p_j$ of $c_j$  must be in $X \cup C$.
Let ${\ell}_{i'} \in c_j$ (possibly, $p_j = {\ell}_{i'}$).
We have $(c_j|v_{i'})_{{\ell}_i} = (2+2-2)/2 = 1$.
So, $\rho_{{\ell}_i,T}(G_{\Phi}) \geq d(v_i,p_j) = 2$.
\end{proof}

\begin{lemma}
For every BFS-tree $T$ rooted at $v_i\in V$, we have $\rho_{v_i,T}(G_{\Phi}) \geq 2$.
\end{lemma}

\begin{proof}
There exists $i' \neq i$ such that ${\ell}_i,{\ell}_{i'}$ are nonadjacent.
In particular, the parent of ${\ell}_{i'}$  must be $v_{i'}$.
Furthermore, there exists $c_j \in C$ such that $d(v_i,c_j) = 2$.
In particular, ${\ell}_i$ must be the parent of $c_j$.
However, $({\ell}_{i'}|c_j)_{v_i} = 2 - d({\ell}_{i'},c_j)/2 \in \{1,3/2\}$.
In particular, $\lfloor ({\ell}_{i'}|c_j)_{v_i} \rfloor = 1$.
So, $\rho_{v_i,T}(G_{\Phi}) \geq d(v_{i'},{\ell}_i) = 2$.
\end{proof}

From now on, let $w$ be the basepoint of $T$.
We prove that for most pairs $s$ and $t$, $d(s_t,t_s) \leq 1$ always holds ({i.e.}, regardless whether $\Phi$ is satisfiable).

\begin{lemma}
If $s \in V$ and $t$ is arbitrary, then $d(s_t,t_s) \leq 1$.
\end{lemma}

\begin{proof}
Since $w$ is a simplicial vertex  and $s \in N(w)$, we have $d(w,t) - 1 \leq d(s,t) \leq d(w,t)$, and consequently,
$$(s|t)_w = (d(s,w)+d(t,w)-d(s,t))/2 = 1/2 + (d(t,w)-d(s,t))/2 \in \{1/2,1\}.$$
In particular, $s_t,t_s \in N[w]$, and so, $d(s_t,t_s) \leq 1$ since $w$ is simplicial.
\end{proof}

\begin{lemma}
If $s,t \in X$, then $d(s_t,t_s) \leq 1$.
\end{lemma}

\begin{proof}
We have $(s|t)_w = 2 - d(s,t)/2$.
In particular, $\lfloor (s|t)_w \rfloor \leq 1$.
As a result, $s_t,t_s \in N[w]$, and since $w$ is simplicial, we obtain $d(s_t,t_s) \leq 1$.
\end{proof}

\begin{lemma}
If $s \in X$ and $t \in C$, then $d(s_t,t_s) \leq 1$.
\end{lemma}

\begin{proof}
We have $(s|t)_w = 5/2 - d(s,t)/2 \in \{3/2,2\}$.
In particular, if $d(s,t) = 2$ then $\lfloor (s|t)_w \rfloor = 1$, and so, we are done because  $s_t,t_s \in N[w]$ and $w$ is simplicial.
Otherwise, $d(s,t) = 1$, and so, $(s|t)_w = 2$.
In particular, $s_t = s$ and $s_t,t_s$ are two literals contained in $t$.
The latter implies $d(s_t,t_s) \leq 1$ since a clause cannot contain a literal and its negation.
\end{proof}

Finally, we prove that in order to get $\rho_{w,T}(G_{\Phi}) \leq 1$, a necessary and sufficient condition is that the
parent nodes in $T$ of the clause vertices are pairwise adjacent in $G_{\Phi}$.
By construction, the latter corresponds to a satisfying assignment for $\Phi$.

\begin{lemma}
If  $s,t \in C$, then $s_t,t_s \in X$.
\end{lemma}

\begin{proof}
We have $(s|t)_w = 3 - d(s,t)/2 \in \{2, 5/2 \}$.
In particular, $\lfloor(s|t)_w\rfloor = 2$.
\end{proof}

\subsection*{Acknowledgements}
We are grateful to the referees of the journal and short versions of
the paper for a careful reading and many useful comments and
suggestions. The research of J.C., V.C., and Y.V.  was supported by
ANR project DISTANCIA (ANR-17-CE40-0015).

\bibliographystyle{plainurl}
\bibliography{approx_hyp}

\end{document}